\pdfoutput=1

\documentclass[sigconf]{aamas} 

\usepackage{balance} 
\usepackage{pgfplots}
\newcommand{\axisWidth}{8cm}
\newcommand{\axisHeight}{5.8cm}
\newcommand{\legendFont}{\small}

\newcommand{\Plu}{\textnormal{Plu}}
\newcommand{\TR}{\textnormal{TR}}
\newcommand{\IRV}{\textnormal{IRV}}

\makeatletter
\gdef\@copyrightpermission{
  \begin{minipage}{0.2\columnwidth}
   \href{https://creativecommons.org/licenses/by/4.0/}{\includegraphics[width=0.90\textwidth]{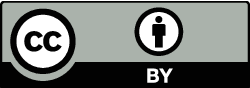}}
  \end{minipage}\hfill
  \begin{minipage}{0.8\columnwidth}
   \href{https://creativecommons.org/licenses/by/4.0/}{This work is licensed under a Creative Commons Attribution International 4.0 License.}
  \end{minipage}
  \vspace{5pt}
}
\makeatother

\setcopyright{ifaamas}
\acmConference[AAMAS '25]{Proc.\@ of the 24th International Conference
on Autonomous Agents and Multiagent Systems (AAMAS 2025)}{May 19 -- 23, 2025}
{Detroit, Michigan, USA}{Y.~Vorobeychik, S.~Das, A.~Nowé  (eds.)}
\copyrightyear{2025}
\acmYear{2025}
\acmDOI{}
\acmPrice{}
\acmISBN{}

\acmSubmissionID{739}

\title[]{Why Instant-Runoff Voting Is \\So Resilient to Coalitional Manipulation: \\Phase Transitions in the Perturbed Culture}

\author{François Durand}
\affiliation{
  \institution{Nokia Bell Labs France}
  \city{Massy}
  \country{France}}
\email{francois.durand@nokia-bell-labs.com}

\begin{abstract}
Previous studies have shown that Instant-Runoff Voting (IRV) is highly resistant to coalitional manipulation (CM), though the theoretical reasons for this remain unclear. To address this gap, we analyze the susceptibility to CM of three major voting rules---Plurality, Two-Round System, and IRV---within the Perturbed Culture model. Our findings reveal that each rule undergoes a phase transition at a critical value \(\theta_c\) of the concentration of preferences: the probability of CM for large electorates converges exponentially fast to 1 below \(\theta_c\) and to 0 above \(\theta_c\). We introduce the Super Condorcet Winner (SCW), showing that its presence is a key factor of IRV's resistance to coalitional manipulation, both theoretically and empirically. Notably, we use this notion to prove that for IRV, \(\theta_c = 0\), making it resistant to CM with even minimal preference concentration.
\end{abstract}


\newcommand{\BibTeX}{\rm B\kern-.05em{\sc i\kern-.025em b}\kern-.08em\TeX}

\begin{document}


\pagestyle{fancy}
\fancyhead{}


\maketitle 


\section{Introduction}

\subsection{Motivation}

The Gibbard-Satterthwaite Theorem \cite{gibbard1973manipulation, satterthwaite1975strategyproofness} shows that all non-trivial voting rules are vulnerable to manipulation (strategic voting), even by a single voter. This vulnerability can only worsen when any number of voters with aligned interests can form a coalition to alter the election outcome, a phenomenon called \emph{coalitional manipulation} (CM). Unlike individual manipulation \cite{peleg1979largevotingschemes, slinko2002classical}, CM remains significant in large-scale elections and raises several concerns, notably creating moral dilemmas \cite[Introduction]{durand2015towards} and power imbalances between strategic and naive voters, thus undermining the "one person, one vote" principle \cite{durand2023coalitional,eggers2020votes}.

However, not all voting rules are equally vulnerable: the \emph{CM rate}, i.e., the probability that a voting profile is manipulable by a coalition under a given probabilistic model, can vary significantly between rules. In previous studies, \emph{Instant-Runoff Voting} (IRV) and some of its variants \cite{durand2023coalitional} consistently outperform other classical single-winner voting rules in resisting coalitional manipulation, whether the analysis is based on randomly generated profiles \cite{green2014strategic, green2016statistical} \cite[Chapters 7--8]{durand2015towards} or experimental datasets \cite{chamberlin1984observed, green2016statistical} \cite[Chapter 9]{durand2015towards}. This has been confirmed by theoretical calculations in the case of three candidates \cite{lepelley1994vulnerability, lepelley2003homogeneity}. This is especially intriguing, as IRV has several theoretical features typically deemed undesirable and seemingly prone to manipulation: most notably, IRV is neither \emph{Condorcet-consistent} nor \emph{monotonic} \cite[Definitions 2.8 and 2.10]{brandt2016handbook} \cite{durand2016condorcet}.

In an effort to shed theoretical light on this phenomenon, we will compare IRV to the two other most widely used voting rules in large-scale single-winner political elections: Plurality and the Two-Round System.
We will adopt the \emph{Perturbed Culture} model of random voting profiles, first introduced by \citet{williamson1967social} and later named by \citet[Section 4.3.2]{gehrlein2006condorcet}. 
We will focus on the asymptotic behavior as the number of voters tends to infinity, as it offers more mathematical tractability and is relevant for large-scale elections. We will also examine convergence rates to assess how well this limit approximates scenarios with finite electorates. Our approach is similar to the use of the Ising model in physics, which, despite being unrealistic in its microscopic details, has been remarkably effective in explaining the complex macroscopic phenomenon of phase transitions in ferromagnetism~\cite{kadanoff2000statistical}.

\subsection{Contributions}

In this paper, we prove that each of the three voting rules undergoes a \emph{phase transition}, with an abrupt change in behavior based on whether the concentration parameter \(\theta\) in the Perturbed Culture model exceeds a critical threshold \(\theta_c\). Below \(\theta_c\), the CM rate tends to 1 for large electorates, while above \(\theta_c\), it tends to 0. We compute the critical threshold \(\theta_c\) for each voting rule as a function of the number of candidates.

We show through simulations how the CM rate curve as a function of \(\theta\), which is continuous for a finite number of voters \(n\), converges to a discontinuous curve as \(n\) tends to infinity, thereby explaining the phase transition. Additionally, we investigate the critical regime \(\theta = \theta_c\), leading to the conjecture that in this case, the CM rate tends to a limit strictly between 0 and 1.

We introduce the concept of a Super Condorcet Winner (SCW), which largely explains IRV's resilience to CM. This leads to one of our most striking results: for IRV, the critical value \(\theta_c\) is 0, regardless of the number of candidates. This means that IRV is asymptotically resistant to CM as soon as the Perturbed Culture model shows even the slightest preference concentration. Furthermore, using experimental datasets, we show that SCWs are frequent in practice and account for most of IRV's resistance to CM.

Finally, we demonstrate that in non-critical regime, i.e., for \mbox{\(\theta \neq \theta_c\)}, the convergence of the CM rate toward 0 or 1 is exponentially fast. This implies that our results for \(n \to \infty\) quickly become relevant even for finite \(n\). We also study how this speed varies with~\( \theta \).

\subsection{Related Work}

In addition to the literature already mentioned, the works most closely related to ours are those that study the CM rate using theoretical tools. Most of them focus on three-candidate elections within models like \emph{Impartial Culture} \cite{lepelley1999kimroush}, \emph{Impartial Anonymous Culture} \cite{lepelley1994vulnerability, favardin2002bordacopeland}, or \emph{Pólya-Eggenberger Urns} \cite{lepelley2003homogeneity}. Although they consistently show that IRV is more robust than other rules, they are limited to a specific number of candidates and offer little intuition for IRV's superior performance. \citet{kim1996manipulability} provide key results for large electorates under \emph{Impartial Culture} for Plurality and some other rules (\emph{positional scoring rules} in general, \emph{Maximin}, and \emph{Coombs}) but do not address the Two-Round System or IRV.

Concerning phase transitions, research on this subject is abundant in physics (see \citet{kadanoff2000statistical} for an overview) and in mathematics and computer science
 \cite{christensen2002percolation, bollobas2006percolation, duminil2018sixty}.
In voting theory, \citet{mossel2013smooth} and \citet{xia2022impact} also examine phase transitions in coalitional manipulability, focusing on varying numbers of manipulators. In contrast, our study considers the impact of the concentration parameter in the probabilistic distribution of preferences.

\subsection{Limitations}

The limitations of this work stem from its main assumptions. First, while the Perturbed Culture model is useful, it does not capture the full complexity of real-world preferences. Second, our analysis is limited to three voting rules, and extending this to other systems would be valuable. Finally, the concept of coalitional manipulation may face criticism due to coordination challenges or the lack of binding agreements among coalition members (see \citet[Introduction]{durand2015towards} for a response to these critiques).

\subsection{Roadmap}

The rest of the paper is organized as follows. Section~\ref{sec_definitions_and_notations} introduces key definitions and notations. Sections \ref{sec_plurality}, \ref{sec_two_round_system}, and \ref{sec_irv} respectively analyze Plurality, Two-Round System, and IRV. Section \ref{sec_convergence_speed} explores convergence speed. Section \ref{sec_conclusion} concludes with future work.

\section{Definitions and Notations}\label{sec_definitions_and_notations}

\subsection{Discrete and Continuous Profiles}

A \emph{discrete profile} \( P \) consists of three elements: a finite, non-empty set of candidates \( \mathcal{C}(P) \) with cardinality \( m(P) \); a finite, non-empty set of voters \( \mathcal{V}(P) \) with cardinality \( n(P) \); and for each voter \( v \in \mathcal{V}(P) \), a preference ranking \( P_v \) over the candidates in \( \mathcal{C}(P) \). 

For any preference ranking \( p \), let \( w(p, P) \) denote the \emph{weight} of~\( p \) in \( P \), i.e., the number of voters in \( P \) with ranking \( p \). The total weight of a discrete profile is simply the number of voters: \( w(P) = \sum_p w(p, P) = n(P) \).

A \emph{continuous profile} is similarly defined by three components: a finite, non-empty set of candidates \( \mathcal{C}(P) \) with cardinality \( m(P) \); a total weight \( w(P) \in (0, \infty) \); and for each ranking \( p \) over the candidates, a weight \( w(p, P) \in \mathbb{R} \), such that \( \sum_p w(p, P) = w(P) \).

For any profile \( P \), whether discrete or continuous, we define the associated \emph{normalized profile} \( \bar{P} \) as the continuous profile where the weight of each ranking \( p \) is given by \( w(p, \bar{P}) = \frac{w(p, P)}{w(P)} \). Viewing a profile as a vector of weights, we can naturally define its \emph{neighborhood} in the usual topological sense.

For any subset \( K \subseteq \mathcal{C}(P) \), let \( P_{K} \) be the restriction of \( P \) to the candidates in \( K \). For two distinct candidates \( c \) and \( d \), let \( P^{c \succ d} \) be the restriction to voters who prefer \( c \) over \( d \). Similarly, for a candidate~\( c \) and a position \( k \in \{1, \ldots, m(P)\} \), let \( P^{r(c) = k} \) be the restriction to voters ranking \( c \) in the \( k \)-th position. These notations can be combined to restrict the profile both by candidates and voters.

\subsection{Voting Rules}

A \emph{voting rule} \( f \) maps any profile, discrete or continuous, to a candidate from that profile. In this paper, we focus on \emph{homogeneous} voting rules, meaning that \( f(P) = f(\bar{P}) \) for any profile \( P \). In other words, the outcome depends only on the relative proportions of preference rankings, not the total weight.
Each particular voting rule is formally defined at the beginning of its respective section.

\subsection{Coalitional Manipulability}

When \( P \) is a discrete profile, we say that a voting rule \( f \) is \emph{coalitionally manipulable} (also abbreviated as CM) in \( P \), or that profile $P$ is CM in rule $f$, if there exists a \emph{target profile} \( Q \) with the same candidates and voters such that \( f(Q) \neq f(P) \), and for every voter \( v \in \mathcal{V}(P) \), if \( Q_v \neq P_v \), then \( v \) prefers \( f(Q) \) to \( f(P) \) based on \( P_v \). In other words, only voters who benefit from the new outcome may alter their ballots, though some may keep their original votes. 

An immediate consequence is as follows. If for a ranking \( p \), we have \( w(p, Q) < w(p, P) \), then at least one voter with ranking \( p \) in \( P \) must have changed their ballot in \( Q \), i.e., \( Q_v \neq P_v \). By the definition, this implies that \( f(Q) \) is preferred to \( f(P) \) according to \( P_v = p \). This observation will now serve as the basis for defining CM in the continuous case.

For a continuous profile \( P \), we say that a voting rule \( f \) is CM in~\( P \) (or that \( P \) is CM in \( f \)) if there exists a target profile \( Q \) with the same candidates and total weight such that \( f(Q) \neq f(P) \) and, for every ranking \( p \), if \( w(p, Q) < w(p, P) \), then \( f(Q) \) is preferred to \( f(P) \) according to \( p \). In other words, only voters (in a continuous sense) who prefer the new outcome can have changed their ballots.

The relationship between the two notions is clarified by:

\begin{lemma}\label{lem_cm_discrete_implies_cm_continuous}
If a homogeneous rule \( f \) is CM in a discrete profile \( P \), then \( f \) is also CM in the corresponding normalized profile \( \bar{P} \). However, the converse is not true.
\end{lemma}

The direct implication follows from the definitions, so we will focus on providing a counterexample to show that the converse does not hold. Consider the \emph{positional scoring rule} \( f \) with weights \( (7, 6, 0, \ldots, 0) \), where each candidate's score is given by \( s(c) = 7 w(P^{r(c) = 1}) + 6 w(P^{r(c) = 2}) \), and the candidate with the highest score wins (using a tie-breaking rule if needed). Now, consider a discrete profile \( P \) with 8 voters and 3 candidates:
\begin{itemize}
	\item 3 voters have the ranking \( 1 \succ 3 \succ 2 \),
	\item 5 voters have the ranking \( 2 \succ 1 \succ 3 \).
\end{itemize}
It is straightforward to verify that candidate~1 wins under \( f \), and that the rule is CM in the normalized profile \( \bar{P} \), but not in the original discrete profile \( P \).
The issue is that the 5 manipulators supporting candidate~2 must carefully distribute their points between candidates~1 and~3
, which is impossible in the discrete case because each manipulator must assign their entire vote to one ranking rather than splitting it fractionally.

\subsection{Perturbed Culture}

Given two positive integers \( m \) and \( n \), and a \emph{concentration parameter} \( \theta \in (0, 1] \), the \emph{Perturbed Culture} model is defined as follows. A discrete profile \( P \) is randomly generated with \( \mathcal{C}(P) = \{1, \ldots, m\} \) and \( \mathcal{V}(P) = \{1, \ldots, n\} \). Each voter is independently assigned the ranking \( (1 \succ \ldots \succ m) \) with probability \( \theta \), and a uniformly random ranking with probability \( 1 - \theta \).

As \( \theta \to 0 \), this model converges to the classical \emph{Impartial Culture} model, while for \( \theta = 1 \), it becomes a deterministic culture where all voters share the ranking \( (1 \succ \ldots \succ m) \).\footnote{We exclude the case \( \theta = 0 \) from our theoretical analysis to simplify the proofs: as \( n \to \infty \), assuming \( \theta > 0 \) guarantees that candidate~1 wins under sincere voting for all three voting rules considered. Nevertheless, our results hold even in the case \( \theta = 0 \), and we will also include it in our figures.}

Since a profile can be represented as a vector giving the weight of each ranking, we can define the \emph{expected normalized profile} (or simply the \emph{expected profile}) under Perturbed Culture. To simplify notation, we denote it by \( \hat{P} \), leaving its dependency on \( m \) and \( \theta \) implicit. In this profile, the ranking \( (1 \succ \ldots \succ m) \) has a weight of \( \theta + \frac{1 - \theta}{m!} \), while each of the other rankings has a weight of \( \frac{1 - \theta}{m!} \).

\subsection{CM Rate}

We denote by \( \rho(f, m, n, \theta) \) the \emph{CM rate}, i.e., the probability that a voting rule \( f \) is CM in a profile drawn from the Perturbed Culture model with $m$ candidates, $n$ voters and concentration $\theta$.

\section{Plurality}\label{sec_plurality}

We will begin our study with the \emph{Plurality} voting rule, which assigns each candidate \( c \) in a profile \( P \) a score equal to the total weight of voters ranking \( c \) first: \( s_\Plu{}(c, P) = w( P^{r(c) = 1} ) \). The winner is the candidate with the highest score (using a tie-breaking rule if needed): \( \Plu{}(P) = \arg\max s_\Plu{}(c, P) \). The specific tie-breaking method will not affect our findings.

\subsection{Theoretical Results for Plurality}

The intuition behind our theoretical results is as follows. First, we analyze Plurality’s behavior in the expected normalized profile~\( \hat{P} \) as a function of \( \theta \). For small \( \theta \), Plurality is CM in this profile, but for large enough \( \theta \), it is not. Using the weak law of large numbers, we then show that as \( n \to \infty \), the normalized random profile~\( \bar{P} \) will, with high probability (i.e., with a probability that tends to 1 when \(n \to \infty\)), be close enough to~\( \hat{P} \), ensuring that Plurality behaves similarly. Throughout this subsection, we assume \( m \geq 2 \).

We begin by analyzing the expected profile \( \hat{P} \). In this profile, the plurality score for candidate 1 is \( s_{\Plu{}}(1, \hat{P}) = \theta + \frac{1 - \theta}{m} \), while the number of voters inclined to manipulate for any candidate \( c \neq 1 \) is \( w(\hat{P}^{c \succ 1}) = \frac{1 - \theta}{2} \). If all manipulators vote optimally for \( c \), they succeed if \( w(\hat{P}^{c \succ 1}) > s_{\Plu{}}(1, \hat{P}) \), which simplifies to \( \theta < \frac{m - 2}{3 m - 2} \). Defining the \emph{critical value} \( \theta_c(\Plu{}, m) = \frac{m - 2}{3 m - 2} \), we conclude that Plurality is CM for \( \theta < \theta_c(\Plu{}, m) \) and not CM for \( \theta > \theta_c(\Plu{}, m) \) (the equality case is not needed for our forthcoming analysis).

We now apply the weak law of large numbers to show that as \( n \to \infty \), these results hold with high probability. We start by examining the \emph{supercritical regime} \( \theta > \theta_c(\Plu{}, m) \), relying on the following lemma.

\begin{lemma}\label{lem_not_CM_in_neighborhood}
	Assume there exists a neighborhood of the expected normalized profile \( \hat{P} \) where the homogeneous rule \( f \) is not CM. Then \( \lim_{n \to \infty} \rho(f, m, n, \theta) = 0 \).
\end{lemma}

\begin{proof}
	Applying the weak law of large numbers, the following statements hold with high probability: denoting $P$ the random profile, its normalized version \( \bar{P} \) lies in the desired neighborhood of~\( \hat{P} \), hence (by assumption) $f$ is not CM in \( \bar{P} \), hence (by Lemma~\ref{lem_cm_discrete_implies_cm_continuous}) $f$ is also not CM in the random discrete profile \( P \).
\end{proof}

This lemma applies easily to Plurality. For \( \theta > \theta_c(\Plu{}, m) \), we have shown that for every candidate \( c \neq 1 \), \( w(\hat{P}^{c \succ 1}) < s_{\Plu{}}(1, \hat{P}) \). As this is a strict inequality, it holds in a neighborhood of the profile, allowing us to apply Lemma~\ref{lem_not_CM_in_neighborhood}. Hence, \( \lim_{n \to \infty} \rho(\Plu{}, m, n, \theta) = 0 \).

We now turn to the \emph{subcritical regime} \( \theta < \theta_c(\Plu{}, m) \). Unfortunately, we cannot directly apply the same reasoning: even if the normalized profile \( \bar{P} \) is CM near \( \hat{P} \), it does not necessarily follow that the discrete profile \( P \) is also CM, as Lemma~\ref{lem_cm_discrete_implies_cm_continuous} does not hold in the reverse direction.

However, for Plurality, manipulators can always employ a common strategy. Formally, a voting rule \( f \) is \emph{unison-manipulable} (UM) in profile \( P \) (or \( P \), in \( f \)) if manipulation can succeed even when all interested voters cast the same ballot~\cite{walsh2010empirical, durand2023coalitional}.\footnote{The term \emph{unison} was introduced by \citet{walsh2010empirical} but we follow the slightly different definition proposed by \citet{durand2023coalitional}.} Clearly, UM implies CM. Unlike CM, UM holds equivalently for both a discrete profile~\( P \) and its normalized profile \( \bar{P} \), which leads to the following lemma.

\begin{lemma}\label{lem_um_in_neighborhood}
	Assume there exists a neighborhood of the expected normalized profile $\hat{P}$ where the homogeneous rule $f$ is UM. Then \( \lim_{n\to\infty} \rho(f, m, n, \theta) = 1 \).
\end{lemma}

The proof is similar to Lemma~\ref{lem_not_CM_in_neighborhood}: by the weak law of large numbers, with high probability, the normalized random profile \( \bar{P} \) is in the desired neighborhood, making it UM, hence the random discrete profile \( P \) is also UM, and thus CM. Applied to Plurality, Lemma~\ref{lem_um_in_neighborhood} directly leads to \( \lim_{n\to\infty} \rho(f, m, n, \theta) = 1 \) for \( \theta < \theta_c(\Plu{}, m) \).

The following theorem summarizes our results so far.

\begin{theorem}\label{thm_plurality_theta_c}
	Let \( \theta_c(\Plu{}, m) = \frac{m - 2}{3 m - 2} \) with \( m \geq 2 \).
	\begin{itemize}
		\item If \( \theta < \theta_c(\Plu{}, m) \), then \( \lim_{n \to \infty} \rho(\Plu{}, m, n, \theta) = 1 \).
		\item If \( \theta > \theta_c(\Plu{}, m) \), then \( \lim_{n \to \infty} \rho(\Plu{}, m, n, \theta) = 0 \).
	\end{itemize}
\end{theorem}

For \( m = 2 \), the theorem indicates \( \theta_c(\Plu{}, 2) = 0 \), which is expected, as Plurality cannot be manipulated with only two candidates. Similarly, for \( m = 1 \), we would reach the same conclusion by conventionally setting \( \theta_c(\Plu{}, 1) = 0 \). The theorem becomes more interesting for \( m \geq 3 \), where it describes a \emph{phase transition} around \( \theta_c(\Plu{}, m) \), meaning a sudden change in behavior as the parameter crosses this threshold. This raises key questions: What causes this discontinuity, and how do we approach it as \( n \) increases? What happens when \( \theta \) is equal to or near the critical value?

\subsection{Simulations for Plurality}

\begin{figure}
	\centering
	\input{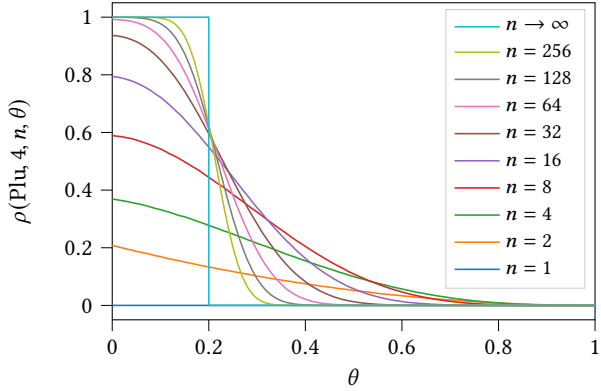}
	\caption{CM rate of Plurality as a function of \( \theta \) for different values of \( n \) with \( m = 4 \). Curves for finite \( n \) are based on Monte Carlo simulations with 1,000,000 profiles per point. The limiting curve as \( n \to \infty \) follows from Theorem~\ref{thm_plurality_theta_c}.}
	\label{fig_plu_cm_rate_of_theta_and_n_v}
	\Description{CM rate of Plurality as a function of theta for different values of n with m = 4. Curves for finite n are based on Monte Carlo simulations with 1,000,000 profiles per point. The limiting curve as n tends to infinity follows from our theoretical results.}
\end{figure}

To understand the origin of the discontinuity, Figure~\ref{fig_plu_cm_rate_of_theta_and_n_v} shows the CM rate of Plurality as a function of \( \theta \) for various \( n \) with \( m = 4 \). Curves for finite \( n \) are based on Monte Carlo simulations,\footnote{The code is available at \url{https://github.com/francois-durand/irv-cm-aamas-2025}.} with 1,000,000 profiles per point, leading to error margins of \( \frac{1}{\sqrt{1000000}} = 0.1\% \). The limiting curve for \( n \to \infty \) is derived from Theorem~\ref{thm_plurality_theta_c}. For finite \( n \), the curve is continuous.  As \( n \) increases, it becomes sigmoid-shaped and steepens, ultimately converging to a step function as \( n \to \infty \). 

The observed behavior mirrors what occurs in physics: since a finite combination of continuous functions remains continuous, non-analyticity can only arise in an infinite system \cite[Section 11.6]{kadanoff2000statistical}. 
As in physics, a phase transition occurs beyond a certain level of disorder: while a ferromagnetic metal loses its magnetization above the Curie temperature \cite{chikazumi1997physics}, Plurality loses its resistance to coalitional manipulation below the critical value of the concentration parameter~\( \theta \).

\begin{figure}
	\centering
\begin{tikzpicture}

\definecolor{crimson2143940}{RGB}{214,39,40}
\definecolor{darkgrey176}{RGB}{176,176,176}
\definecolor{darkorange25512714}{RGB}{255,127,14}
\definecolor{forestgreen4416044}{RGB}{44,160,44}
\definecolor{lightgrey204}{RGB}{204,204,204}
\definecolor{steelblue31119180}{RGB}{31,119,180}

\begin{axis}[reverse legend,
height=\axisHeight,
legend cell align={left},
legend style={font=\legendFont, fill opacity=1, draw opacity=1, text opacity=1, draw=lightgrey204},
legend pos=south east,
log basis x={10},
tick align=outside,
tick pos=left,
width=\axisWidth,
x grid style={darkgrey176},
xlabel={$n$},
xmin=1.0, xmax=1000,
xmode=log,
xtick style={color=black},
xmode=log,
y grid style={darkgrey176},
ylabel={$\rho(\text{Plu}, m, n, \theta_c(\text{Plu}, m))$},
ymin=-0.05, ymax=1.05,
ytick={-0.2, 0.0, 0.2, 0.4000000000000001, 0.6000000000000001, 0.8, 1.0000000000000002, 1.2000000000000002},
ytick style={color=black}
]
\addplot [semithick, steelblue31119180]
table {%
1 0
2 0.081567
3 0.157474
4 0.170542
5 0.231661
6 0.244106
7 0.281132
8 0.299572
9 0.32236
10 0.338951
11 0.356247
12 0.370793
13 0.384525
14 0.397044
15 0.408729
16 0.419106
17 0.429511
18 0.437013
19 0.445895
20 0.453346
21 0.461071
22 0.467627
23 0.473501
24 0.478436
25 0.484474
26 0.490005
27 0.495122
28 0.497907
29 0.502796
30 0.505565
31 0.508733
32 0.512484
33 0.516416
34 0.519258
35 0.521766
36 0.52428
37 0.528044
38 0.528772
39 0.531339
40 0.533469
41 0.53488
42 0.537406
43 0.538839
44 0.540549
45 0.54177
46 0.543557
47 0.543438
48 0.546643
49 0.547347
50 0.549067
51 0.549746
52 0.550749
53 0.551627
54 0.552769
55 0.552338
56 0.554126
57 0.555042
58 0.555698
59 0.556517
60 0.558622
61 0.55831
62 0.558445
63 0.560072
64 0.559491
65 0.560674
66 0.56052
67 0.561517
68 0.561529
69 0.56244
70 0.562275
71 0.563731
72 0.563665
73 0.564071
74 0.565129
75 0.564629
76 0.565405
77 0.565702
78 0.566221
79 0.566172
80 0.565782
81 0.567196
82 0.567672
83 0.567459
84 0.567885
85 0.567548
86 0.567685
87 0.567664
88 0.568419
89 0.568254
90 0.568532
91 0.569585
92 0.569407
93 0.569236
94 0.570191
95 0.569031
96 0.570267
97 0.569737
98 0.571005
99 0.570016
100 0.570672
125 0.573146
158 0.575854
199 0.576767
251 0.579197
316 0.581613
398 0.58257
501 0.583311
630 0.583523
794 0.585407
1000 0.585653
};
\addlegendentry{$m=3$}
\addplot [semithick, darkorange25512714]
table {%
1 0
2 0.133102
3 0.266636
4 0.277303
5 0.367169
6 0.381135
7 0.425513
8 0.444741
9 0.467914
10 0.485436
11 0.500954
12 0.512965
13 0.523955
14 0.534406
15 0.542132
16 0.549401
17 0.554261
18 0.560205
19 0.565782
20 0.569755
21 0.573558
22 0.576145
23 0.579635
24 0.581824
25 0.584952
26 0.58646
27 0.588477
28 0.591391
29 0.59196
30 0.592951
31 0.594191
32 0.595813
33 0.597277
34 0.598192
35 0.598797
36 0.598512
37 0.600891
38 0.6013
39 0.601714
40 0.601336
41 0.602472
42 0.603605
43 0.603879
44 0.603836
45 0.605102
46 0.605902
47 0.605631
48 0.605706
49 0.606978
50 0.607432
51 0.607671
52 0.607802
53 0.607968
54 0.60776
55 0.608975
56 0.609556
57 0.608385
58 0.609783
59 0.61007
60 0.609764
61 0.610707
62 0.610792
63 0.610005
64 0.610994
65 0.611078
66 0.612219
67 0.612541
68 0.611979
69 0.61232
70 0.612667
71 0.612377
72 0.612546
73 0.612871
74 0.613391
75 0.613271
76 0.613523
77 0.613343
78 0.613091
79 0.61384
80 0.613798
81 0.61441
82 0.615611
83 0.614557
84 0.613602
85 0.615208
86 0.615295
87 0.614794
88 0.615266
89 0.615984
90 0.615322
91 0.616602
92 0.616576
93 0.616212
94 0.616277
95 0.616186
96 0.616767
97 0.616509
98 0.616263
99 0.617624
100 0.617238
125 0.618635
158 0.622302
199 0.624305
251 0.624923
316 0.626913
398 0.627985
501 0.629958
630 0.631392
794 0.631688
1000 0.633079
};
\addlegendentry{$m=4$}
\addplot [semithick, forestgreen4416044]
table {%
1 0
2 0.169113
3 0.340604
4 0.346674
5 0.448318
6 0.457756
7 0.502352
8 0.518208
9 0.538645
10 0.553015
11 0.564087
12 0.573922
13 0.581885
14 0.588744
15 0.59399
16 0.598356
17 0.603559
18 0.605652
19 0.608385
20 0.611342
21 0.613112
22 0.615645
23 0.617325
24 0.6185
25 0.619911
26 0.621653
27 0.622837
28 0.623586
29 0.624266
30 0.625324
31 0.626569
32 0.627042
33 0.628494
34 0.627786
35 0.629501
36 0.629546
37 0.630784
38 0.63129
39 0.63056
40 0.630548
41 0.632872
42 0.632649
43 0.633309
44 0.633606
45 0.633769
46 0.634564
47 0.634367
48 0.634886
49 0.636513
50 0.636575
51 0.636299
52 0.637382
53 0.637485
54 0.637115
55 0.6373
56 0.637663
57 0.638735
58 0.638179
59 0.639638
60 0.638782
61 0.638546
62 0.639897
63 0.639696
64 0.640905
65 0.639923
66 0.640982
67 0.64145
68 0.640924
69 0.641406
70 0.640841
71 0.641322
72 0.64187
73 0.64206
74 0.641464
75 0.642153
76 0.642778
77 0.642921
78 0.64229
79 0.642634
80 0.643752
81 0.642979
82 0.64367
83 0.64403
84 0.643675
85 0.644082
86 0.644164
87 0.643589
88 0.644699
89 0.645103
90 0.645244
91 0.645601
92 0.645629
93 0.644692
94 0.64582
95 0.644872
96 0.645595
97 0.645891
98 0.645865
99 0.646131
100 0.645187
125 0.648289
158 0.651416
199 0.652919
251 0.654452
316 0.655854
398 0.657953
501 0.658744
630 0.660527
794 0.662524
1000 0.662682
};
\addlegendentry{$m=5$}
\addplot [semithick, crimson2143940]
table {%
1 0
2 0.196763
3 0.393751
4 0.392037
5 0.500047
6 0.50518
7 0.548344
8 0.560674
9 0.578924
10 0.590423
11 0.598692
12 0.606506
13 0.612189
14 0.617433
15 0.621883
16 0.624395
17 0.62715
18 0.630035
19 0.632771
20 0.634894
21 0.63638
22 0.637567
23 0.639003
24 0.639647
25 0.641017
26 0.642403
27 0.643835
28 0.643415
29 0.644935
30 0.646141
31 0.647292
32 0.647862
33 0.648216
34 0.649344
35 0.650006
36 0.650556
37 0.650659
38 0.651853
39 0.652524
40 0.651836
41 0.652366
42 0.653203
43 0.654429
44 0.655322
45 0.654731
46 0.655047
47 0.657345
48 0.654903
49 0.656365
50 0.656177
51 0.657446
52 0.656681
53 0.657292
54 0.657106
55 0.658169
56 0.657663
57 0.658991
58 0.659349
59 0.659732
60 0.659692
61 0.660242
62 0.660715
63 0.661214
64 0.661125
65 0.661149
66 0.661472
67 0.661068
68 0.661693
69 0.661577
70 0.662278
71 0.661563
72 0.662575
73 0.663373
74 0.662595
75 0.663355
76 0.663609
77 0.662664
78 0.66344
79 0.663215
80 0.663943
81 0.664955
82 0.664491
83 0.665247
84 0.664473
85 0.664639
86 0.665016
87 0.665614
88 0.665353
89 0.664703
90 0.665634
91 0.665261
92 0.664995
93 0.66545
94 0.666268
95 0.666979
96 0.666064
97 0.666433
98 0.667212
99 0.665987
100 0.666737
125 0.669437
158 0.671924
199 0.673719
251 0.675553
316 0.677678
398 0.678723
501 0.679769
630 0.681499
794 0.682553
1000 0.682757
};
\addlegendentry{$m=6$}
\end{axis}

\end{tikzpicture}
	\caption{CM rate of Plurality as a function of $n$ for different values of $m$ with $\theta = \theta_c(\Plu{}, m)$. Monte Carlo simulations with 1,000,000 profiles per point.}
	\label{fig_plu_cm_rate_at_theta_c_of_n_v_n_c}
	\Description{CM rate of Plurality as a function of n for different values of m with critical theta. Monte Carlo simulations with 1,000,000 profiles per point.}
\end{figure}
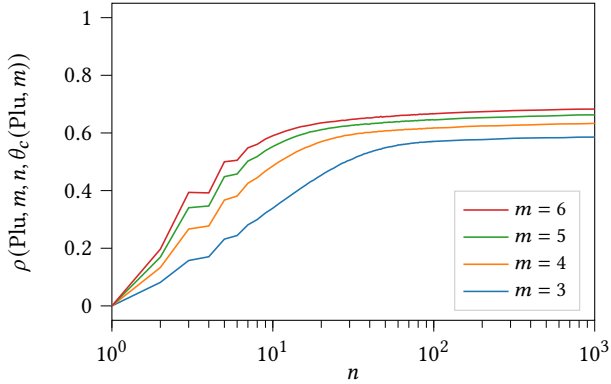

Theorem~\ref{thm_plurality_theta_c} describes the behavior in the subcritical and supercritical regimes, but what happens in the \emph{critical regime}, i.e., when $\theta = \theta_c(\Plu{}, m)$?
Figure~\ref{fig_plu_cm_rate_at_theta_c_of_n_v_n_c} shows the CM rate in that case as a function of the number of voters~\( n \), for different values of \( m \). This leads to several conjectures:
\begin{itemize}
	\item The \emph{critical CM rate} \( \rho(\Plu{}, m, n, \theta_c(\Plu{}, m)) \) converges to a limit as \( n \to \infty \).
	\item This limit is strictly less than 1.
	\item This limit increases with \( m \).
\end{itemize}

It is beyond the scope of this paper to theoretically prove these results. In the study of phase transitions, analyzing the critical behavior is often challenging \cite{newman2000efficient,duminil2017lectures,kos2016precision, komargodski2017random}.
With that, we conclude our study of Plurality and proceed to the Two-Round System.

\section{Two Round System}\label{sec_two_round_system}

The \emph{Two-Round System} (TR) is as follows. In the first round, each candidate \( c \) receives a score \( s^1_\TR{}(c, P) = s_\Plu{}(c, P) \), and the set \( K \) of the two candidates with the highest scores advances to the second round. These two candidates then receive scores \( s^2_\TR{}(c, P) = s_\Plu{}(c, P_K) \), and the candidate with the highest score wins. A tie-breaking rule is applied if necessary.\footnote{For simplicity, we consider an ``instant'' version of TR, where voters cast their ballots once. In most actual implementations, voters participate in two rounds. While this is equivalent for sincere voting, the instant version restricts some manipulation strategies \cite[Table 1.1]{durand2015towards}. However, our results apply to both variants.}

\subsection{Theoretical Results for the Two-Round System}

For \( n = 2 \), the Two-Round System is equivalent to Plurality, so we focus on the case \( m \geq 3 \). As with Plurality, we first analyze the expected profile and then extend the results using the weak law of large numbers. The key difference is of technical nature: unison manipulation is generally insufficient in the Two-Round System.

We begin by examining the expected normalized profile \(\hat{P}\). Candidate 1 clearly wins the election, with the second-round opponent determined by the tie-breaking rule. For a manipulation to succeed in favor of a candidate \( c \neq 1 \), candidate \( c \) must reach the second round. However, if candidate 1 also advances, the manipulation will fail. Therefore, the second round must involve a candidate \( d \notin \{1, c\} \), which is possible since we assumed \( m \geq 3 \). Now, consider the portion of the first-round scores for candidates 1, \(c\), and \(d\) coming from sincere voters:
\[
\left\{\begin{array}{l}
s_\TR^1(1, \hat{P}^{1 > c}) = \theta + \frac{1 - \theta}{m}, \\[2pt]
s_\TR^1(c, \hat{P}^{1 > c}) = 0, \\[2pt]
s_\TR^1(d, \hat{P}^{1 > c}) = \frac{1 - \theta}{2m},
\end{array}\right.
\]  
where, for example, $s_\TR^1(d, \hat{P}^{1 > c})$ denotes the first-round score, in the two-round system, of candidate~$d$ in the restriction of the expected normalized profile $\hat{P}$ to the voters who prefer candidate~1 to candidate~$c$ (i.e., ``sincere'' voters).

For both candidates \( c \) and \( d \) to surpass candidate 1's score, at least \( \theta + \frac{1 - \theta}{m} \) manipulators must vote for \( c \), while \( \theta + \frac{1 - \theta}{m} - \frac{1 - \theta}{2m} \) must vote for \( d \). Therefore, the total number of manipulators, given by \( w(\hat{P}^{c \succ 1}) = \frac{1 - \theta}{2} \), must be at least the sum of these two quantities. Simplifying, the necessary condition becomes \( \theta \leq \frac{m - 3}{5m - 3} \). In other words, coalitional manipulation is impossible for \( \theta > \frac{m - 3}{5m - 3} \).

Let us now show that for \(\theta < \frac{m-3}{5m-3}\), manipulation is possible. Consider the specific case where \( c = 2 \) and \( d = 3 \). Suppose that:
\begin{itemize}
	\item A fraction \(\frac{2(1 - 2\theta)}{3(1 - \theta)}\) of manipulators vote \( (2 \succ \ldots \succ m \succ 1) \),
	\item And a fraction \(\frac{1 + \theta}{3(1 - \theta)}\) of them vote \( (3 \succ \ldots \succ m \succ 1 \succ 2) \).
\end{itemize}
The first fraction is positive since \(\theta < \frac{m - 3}{5m - 3} < \frac{1}{2}\), and the two fractions clearly sum to 1. We chose these values to ensure candidate~2’s lead over candidate~1 in the first round is double that of candidate~3’s, though the reasoning holds for slightly different values. Standard calculations then show that candidates~2 and 3 lead the first round, with candidate~2 winning the second.

In summary, the behavior of the Two-Round System in the expected normalized profile shifts at \(\theta_c(\TR{}, m) = \frac{m-3}{5m-3}\). The next step is to extend this result to discrete profiles as \(n \to \infty\). For the supercritical regime \( \theta > \theta_c(\TR{}, m) \), the reasoning mirrors that for Plurality, relying on Lemma~\ref{lem_not_CM_in_neighborhood}.

In the subcritical regime \( \theta < \theta_c \), the situation is more subtle. Unison manipulation is generally insufficient for the Two-Round System~\cite[Table 1.1]{durand2015towards}, so we cannot directly apply Lemma~\ref{lem_um_in_neighborhood}.

Let us take a step back for the intuition of the problem. Suppose that a rule \(f\) is CM in a neighborhood of the expected profile \(\hat{P}\), and consider a discrete profile \(P\) whose normalized form \(\bar{P}\) lies in this neighborhood, similarly to Lemma~\ref{lem_um_in_neighborhood} for unison manipulation. This means that \(f\) is CM from \(\bar{P}\) toward a target continuous profile \(Q\). The challenge arises from discretization: manipulators in the discrete profile \(P\) may only be able to target a profile whose normalized form approximates \(Q\).

However, as \(n\) grows large, these finite-size effects diminish, allowing manipulators to approach \(Q\) closely. For CM to succeed, it suffices that \(f\) produces the same outcome when a profile is sufficiently close to $Q$. Thus, we will rely on two neighborhoods: one for manipulability near \(\hat{P}\), and another for stability of the outcome around the target profile. Technically, to obtain the desired result, we will require the second neighborhood to be uniform with respect to the first, meaning that for different choices of \(\bar{P}\) near \(\hat{P}\), the neighborhood of \(Q\) that guarantees outcome stability does not shrink arbitrarily. Let us now formalize all this.

\begin{definition}
	Let \(\delta > 0\). We say that a rule \(f\) is \emph{\(\delta\)-stable-CM} in a continuous profile \(P\) (or that \( P \) is \(\delta\)-stable-CM in \( f \)) if there exists a continuous profile \(Q\) such that:
	\begin{itemize}
		\item \(f\) is CM from \(P\) to \(Q\),
		\item For any profile \(Q'\) with \(d(Q, Q') < \delta\), we have \(f(Q') = f(Q)\).
	\end{itemize}
	Here, \(d(Q, Q')\) denotes the \(\ell^\infty\)-distance between profiles, viewed as vectors of weights. However, our results hold for any norm-derived distance, as all norms are equivalent in finite-dimensional spaces.
\end{definition}

\begin{lemma}\label{lem_eta_stable_cm}
	Assume there exists \(\delta > 0\) and a neighborhood of the expected normalized profile \(\hat{P}\) where the homogeneous rule \(f\) is \(\delta\)-stable-CM. Then \(\lim_{n\to\infty} \rho(f, m, n, \theta) = 1\).
\end{lemma}

The uniformity property mentioned earlier lies in that the same value of \(\delta\) applies consistently on the entire neighborhood of \(\hat{P}\).

\begin{proof}
Consider \(n\) large enough such that \(\frac{1}{n} < \delta\). The weak law of large numbers ensures that, with high probability:
\begin{enumerate}
	\item The normalized random profile \(\bar{P}\) lies in the desired neighborhood and is therefore \(\delta\)-stable-CM to a target profile \(Q\).
	\item The manipulators in the discrete profile \(P\) can approach \(Q\) with a margin of one manipulator, causing deviations of at most \(\frac{1}{n} < \delta\) in each normalized weight.
	\item By assumption, this implies that the voting outcome for the approximated target profile matches that of \(Q\).
	\item Therefore, \(f\) is CM in \(P\).\qedhere
\end{enumerate}%
\end{proof}

We can now address the Two-Round System for \(\theta < \frac{m - 3}{5m - 3}\).

Let \(\zeta > 0\) (its value will be specified later). Consider a continuous profile \(P\) such that \(d(P, \hat{P}) < \frac{\zeta}{m!}\). Then, all the scores and score gaps calculated for \(\hat{P}\) still hold for \(P\), with an additive error of at most \(\zeta\). It follows that if \(\zeta < \theta\), candidate 1 wins in \(P\), just as in \( \hat{P} \).

As we did in the case of the expected profile, consider the profile \(Q\) obtained from \(P\) by modifying the ballots of voters who prefer candidate~2 over candidate~1 as follows:
\begin{itemize}
	\item A fraction \(\frac{2(1 - 2\theta)}{3(1 - \theta)}\) vote \((2 \succ \ldots \succ m \succ 1)\),
	\item And a fraction \(\frac{1 + \theta}{3(1 - \theta)}\) vote \((3 \succ \ldots \succ m \succ 1 \succ 2)\).
\end{itemize}
To apply Lemma~\ref{lem_eta_stable_cm}, it suffices to show that, sufficiently close to \(Q\) (including the profile \(Q\) itself), candidate~2 wins.

Let \(\xi > 0\) (specified later) and set \(\delta = \frac{\xi}{m!}\). For a profile \(Q'\) such that \(d(Q', Q) < \delta\), the first-round scores are bounded as follows:
\[
\left\{
\begin{array}{l}
	s_\TR^1(1, Q') < \left(\frac{1 - \theta}{m} + \theta + \zeta\right) + 0 + \xi, \\
	s_\TR^1(2, Q') > 0 + \left( \frac{1 - \theta}{2} - \zeta \right) \frac{2 (1 - 2 \theta)}{3 (1 - \theta)} - \xi, \\
	s_\TR^1(3, Q') > \left( \frac{1 - \theta}{2m} - \zeta \right) + \left( \frac{1 - \theta}{2} - \zeta \right) \frac{1 + \theta}{3 (1 - \theta)} - \xi, \\
	s_\TR^1(c, Q') < \left(\frac{1 - \theta}{2m} + \zeta\right) + 0 + \xi, \quad \text{for any } c > 3.
\end{array}
\right.
\]

Each score above is presented as the sum of three terms: the first comes from sincere voters in \(P^{1 \succ 2}\) and stays within \(\zeta\) of the corresponding value for \(\hat{P}^{1 \succ 2}\) (except for candidate 2, by definition of \(P^{1 \succ 2}\)); the second comes from manipulators in \(Q\); and the third is an error between \(Q\) and \(Q'\), bounded by \(\xi\). If \(\zeta = \xi < \frac{m - 3 - (5m - 3) \theta}{30m}\), standard calculations show that candidates 2 and 3 secure the highest scores and thus advance to the second round.

In the second round, we have:
\[
s_\TR^2(2, Q') - s_\TR^2(3, Q') > \left( \theta - \frac{1 - \theta}{6} - \zeta \right) + \left( \frac{1 - \theta}{2} - \zeta \right) \frac{1 - 5 \theta}{3 (1 - \theta)} - \xi.
\]
If \(\zeta = \xi < \frac{\theta}{12}\), standard calculations show that this quantity is positive, and therefore candidate~2 wins, as desired.

Overall, it thus suffices to take \(\zeta = \xi < \min\left(\frac{m - 3 - (5m - 3) \theta}{30m}, \frac{\theta}{12}\right)\).

We summarize our findings in the following theorem.

\begin{theorem}\label{thm_tr_theta_c}
	Let \( \theta_c(\TR{}, m) = \frac{m-3}{5m-3} \) with \( m \geq 3 \).
	\begin{itemize}
		\item If \( \theta < \theta_c(\TR{}, m) \), then \( \lim_{n \to \infty} \rho(\TR{}, m, n, \theta) = 1 \).
		\item If \( \theta > \theta_c(\TR{}, m) \), then \( \lim_{n \to \infty} \rho(\TR{}, m, n, \theta) = 0 \).
	\end{itemize}
\end{theorem}

Recall that for \( m = 2 \), the same conclusions hold by setting \( \theta_c(\TR{}, 2) = 0 \), since Two-Round is equivalent to Plurality in this case. For \( m = 3 \), the theorem also gives \( \theta_c(\TR{}, 3) = 0 \), which is remarkable: according to the Gibbard-Satterthwaite theorem, manipulability becomes an issue from \( m = 3 \), yet Two-Round avoids this with high probability in Perturbed Culture as soon as \( \theta > 0 \).

Since Theorems \ref{thm_plurality_theta_c} and \ref{thm_tr_theta_c} share similar structures, a natural question arises: does every voting rule \(f\) have a critical parameter \(\theta_c(f, m)\) with similar properties? The answer is no. Consider a rule~\(f\) that uses Plurality when \(n\) is even and Two-Round when \(n\) is odd.\footnote{For simplicity, this counter-example involves a non-homogeneous rule.} From Theorems \ref{thm_plurality_theta_c} and \ref{thm_tr_theta_c}, it follows that for \(\theta < \frac{m-3}{5m-3}\), the CM rate converges to 1, while for \(\theta > \frac{m-2}{3m-2}\), it converges to 0. However, for \(\theta \in \left(\frac{m-3}{5m-3}, \frac{m-2}{3m-2}\right)\), the CM rate tends to 1 for even \(n\) and to 0 for odd \(n\): overall, it does not converge.

It is still possible to define a lower critical value \(\theta_l(f, m)\) and an upper critical value \(\theta_u(f, m)\) as, respectively, the largest and smallest values in \([0, 1]\) such that:
\begin{itemize}
	\item If \( \theta < \theta_l(f, m) \), then \( \lim_{n \to \infty} \rho(f, m, n, \theta) = 1 \),
	\item If \( \theta > \theta_u(f, m) \), then \( \lim_{n \to \infty} \rho(f, m, n, \theta) = 0 \).
\end{itemize}
We can then define \(\theta_c(f, m)\) as their common value when it exists. 
With this convention, Theorems~\ref{thm_plurality_theta_c} and \ref{thm_tr_theta_c} are summarized as:
\[
\theta_c(\Plu{}, m) = \frac{m-2}{3m-2}, \quad \theta_c(\TR{}, m) = \frac{m-3}{5m-3}.
\]

\subsection{Simulations for the Two-Round System}

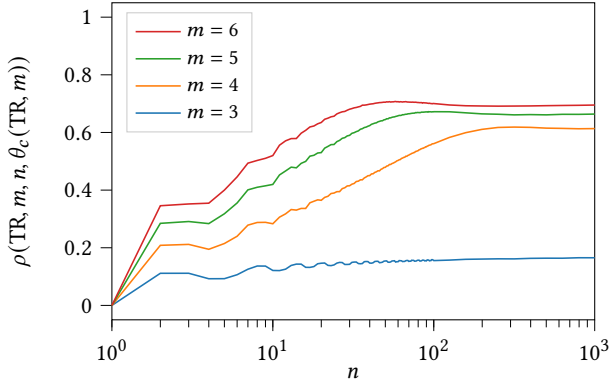
\begin{figure}
	\centering
\begin{tikzpicture}

\definecolor{crimson2143940}{RGB}{214,39,40}
\definecolor{darkgrey176}{RGB}{176,176,176}
\definecolor{darkorange25512714}{RGB}{255,127,14}
\definecolor{forestgreen4416044}{RGB}{44,160,44}
\definecolor{lightgrey204}{RGB}{204,204,204}
\definecolor{steelblue31119180}{RGB}{31,119,180}

\begin{axis}[reverse legend,
height=\axisHeight,
legend cell align={left},
legend style={font=\legendFont, fill opacity=1, draw opacity=1, text opacity=1, draw=lightgrey204},
legend pos=north west,
log basis x={10},
tick align=outside,
tick pos=left,
width=\axisWidth,
x grid style={darkgrey176},
xlabel={$n$},
xmin=1.0, xmax=1000,
xmode=log,
xtick style={color=black},
xmode=log,
y grid style={darkgrey176},
ylabel={$\rho(\text{TR}, m, n, \theta_c(\text{TR}, m))$},
ymin=-0.05, ymax=1.05,
ytick={-0.2, 0.0, 0.2, 0.4000000000000001, 0.6000000000000001, 0.8, 1.0000000000000002, 1.2000000000000002},
ytick style={color=black}
]
\addplot [semithick, steelblue31119180]
table {%
1 0
2 0.111382
3 0.111496
4 0.092498
5 0.092739
6 0.106444
7 0.12551
8 0.136654
9 0.136408
10 0.121208
11 0.120775
12 0.125715
13 0.138624
14 0.143275
15 0.14315
16 0.132097
17 0.131892
18 0.134729
19 0.144495
20 0.146355
21 0.14702
22 0.138876
23 0.137956
24 0.140428
25 0.146692
26 0.149223
27 0.148615
28 0.142047
29 0.14209
30 0.142566
31 0.14935
32 0.150334
33 0.150464
34 0.145243
35 0.14446
36 0.145232
37 0.151645
38 0.152292
39 0.151994
40 0.146935
41 0.146478
42 0.147962
43 0.152396
44 0.15299
45 0.153046
46 0.148225
47 0.148728
48 0.149943
49 0.153426
50 0.154136
51 0.154335
52 0.149853
53 0.149706
54 0.15086
55 0.15404
56 0.154471
57 0.154904
58 0.15035
59 0.151659
60 0.15157
61 0.155506
62 0.15509
63 0.154792
64 0.152148
65 0.152537
66 0.152695
67 0.155048
68 0.155871
69 0.155875
70 0.152725
71 0.153501
72 0.153377
73 0.15672
74 0.156172
75 0.156256
76 0.153438
77 0.153688
78 0.154218
79 0.156389
80 0.156666
81 0.156902
82 0.153853
83 0.154119
84 0.154761
85 0.156953
86 0.157121
87 0.15649
88 0.154399
89 0.15523
90 0.154572
91 0.157758
92 0.157688
93 0.157169
94 0.154681
95 0.155503
96 0.155381
97 0.158563
98 0.158673
99 0.157548
100 0.1555
125 0.157167
158 0.159912
199 0.160733
251 0.161608
316 0.161407
398 0.163238
501 0.163799
630 0.163697
794 0.165343
1000 0.16528
};
\addlegendentry{$m=3$}
\addplot [semithick, darkorange25512714]
table {%
1 0
2 0.20841
3 0.211673
4 0.19467
5 0.214697
6 0.239897
7 0.278462
8 0.287714
9 0.288201
10 0.283259
11 0.306167
12 0.31758
13 0.332563
14 0.330649
15 0.336434
16 0.336371
17 0.355974
18 0.36135
19 0.366224
20 0.364701
21 0.37479
22 0.378555
23 0.388155
24 0.3898
25 0.394567
26 0.394777
27 0.405169
28 0.408173
29 0.414483
30 0.415106
31 0.420543
32 0.423014
33 0.429782
34 0.431058
35 0.435651
36 0.435727
37 0.442194
38 0.445647
39 0.449896
40 0.450357
41 0.456095
42 0.456673
43 0.461428
44 0.463431
45 0.467007
46 0.466826
47 0.472191
48 0.474311
49 0.47837
50 0.479028
51 0.483386
52 0.484347
53 0.488878
54 0.488998
55 0.492173
56 0.493651
57 0.496003
58 0.498276
59 0.501514
60 0.502836
61 0.505319
62 0.506052
63 0.509934
64 0.511093
65 0.51344
66 0.513938
67 0.51749
68 0.517219
69 0.520177
70 0.521435
71 0.523657
72 0.525088
73 0.52737
74 0.528146
75 0.530026
76 0.530588
77 0.533924
78 0.534512
79 0.535866
80 0.536834
81 0.539796
82 0.540411
83 0.542731
84 0.543442
85 0.545126
86 0.545112
87 0.547029
88 0.547906
89 0.550253
90 0.550127
91 0.552524
92 0.553095
93 0.555073
94 0.555797
95 0.556559
96 0.557117
97 0.559196
98 0.56055
99 0.560808
100 0.561139
125 0.582728
158 0.599295
199 0.610663
251 0.617433
316 0.618647
398 0.617702
501 0.614986
630 0.614086
794 0.612839
1000 0.613427
};
\addlegendentry{$m=4$}
\addplot [semithick, forestgreen4416044]
table {%
1 0
2 0.284688
3 0.291012
4 0.283772
5 0.317898
6 0.355407
7 0.400832
8 0.409862
9 0.414868
10 0.419734
11 0.452607
12 0.466432
13 0.479621
14 0.477313
15 0.492804
16 0.499594
17 0.517789
18 0.520784
19 0.528702
20 0.530134
21 0.545618
22 0.551641
23 0.557786
24 0.558243
25 0.568198
26 0.572578
27 0.579758
28 0.582049
29 0.588033
30 0.590023
31 0.597067
32 0.59926
33 0.602899
34 0.60545
35 0.611141
36 0.613153
37 0.617009
38 0.618543
39 0.621634
40 0.623846
41 0.628257
42 0.63024
43 0.632361
44 0.633456
45 0.636826
46 0.638479
47 0.640212
48 0.640602
49 0.643706
50 0.644664
51 0.647539
52 0.647365
53 0.650539
54 0.650718
55 0.652572
56 0.653573
57 0.654062
58 0.655613
59 0.656564
60 0.657879
61 0.659188
62 0.658968
63 0.659665
64 0.661659
65 0.661382
66 0.662468
67 0.663219
68 0.663226
69 0.66459
70 0.664943
71 0.665073
72 0.665912
73 0.666419
74 0.666822
75 0.666995
76 0.667105
77 0.668894
78 0.668359
79 0.669372
80 0.668861
81 0.669266
82 0.668522
83 0.66946
84 0.670527
85 0.67045
86 0.670224
87 0.670844
88 0.670724
89 0.671332
90 0.671212
91 0.671613
92 0.670523
93 0.672124
94 0.671538
95 0.671822
96 0.671621
97 0.671195
98 0.671773
99 0.671458
100 0.672011
125 0.671796
158 0.66813
199 0.664515
251 0.663196
316 0.661058
398 0.661958
501 0.661095
630 0.662965
794 0.662815
1000 0.663761
};
\addlegendentry{$m=5$}
\addplot [semithick, crimson2143940]
table {%
1 0
2 0.345745
3 0.351945
4 0.354638
5 0.40003
6 0.445931
7 0.493492
8 0.502916
9 0.51031
10 0.519874
11 0.555821
12 0.57094
13 0.578691
14 0.578999
15 0.597388
16 0.60737
17 0.619077
18 0.622436
19 0.630004
20 0.634614
21 0.645774
22 0.649545
23 0.654265
24 0.656357
25 0.664373
26 0.668466
27 0.671652
28 0.672959
29 0.677199
30 0.680631
31 0.682838
32 0.68539
33 0.686684
34 0.688739
35 0.690451
36 0.694702
37 0.69471
38 0.695532
39 0.696855
40 0.698228
41 0.698975
42 0.699584
43 0.701015
44 0.70199
45 0.702691
46 0.703596
47 0.70379
48 0.704418
49 0.705125
50 0.704819
51 0.706317
52 0.705721
53 0.706785
54 0.705491
55 0.706248
56 0.706152
57 0.707018
58 0.707416
59 0.706375
60 0.70689
61 0.705704
62 0.706664
63 0.70623
64 0.705619
65 0.7064
66 0.70634
67 0.706942
68 0.705975
69 0.705492
70 0.705481
71 0.705153
72 0.705201
73 0.704929
74 0.705337
75 0.704863
76 0.704757
77 0.704972
78 0.704592
79 0.703953
80 0.703963
81 0.70354
82 0.703395
83 0.703496
84 0.703536
85 0.701903
86 0.703222
87 0.70283
88 0.702684
89 0.70246
90 0.702347
91 0.701713
92 0.701426
93 0.700846
94 0.70142
95 0.700729
96 0.700334
97 0.700056
98 0.699658
99 0.700471
100 0.700304
125 0.695388
158 0.692641
199 0.691839
251 0.691296
316 0.691618
398 0.692046
501 0.692943
630 0.693646
794 0.694206
1000 0.695115
};
\addlegendentry{$m=6$}
\end{axis}

\end{tikzpicture}
	\caption{CM rate of the Two-Round System as a function of~$n$ for different values of $m$ with $\theta = \theta_c(\TR{}, m)$. Monte Carlo simulations with 1,000,000 profiles per point.}
	\label{fig_tr_cm_rate_at_theta_c_of_n_v_n_c}
	\Description{CM rate of the Two-Round System as a function of n for different values of m with critical theta. Monte Carlo simulations with 1,000,000 profiles per point.}
\end{figure}

The Two-Round equivalent of Figure~\ref{fig_plu_cm_rate_of_theta_and_n_v} is similar, so we proceed directly to the counterpart of Figure~\ref{fig_plu_cm_rate_at_theta_c_of_n_v_n_c}: Figure~\ref{fig_tr_cm_rate_at_theta_c_of_n_v_n_c}, showing the critical CM rate as a function of \(n\) for different \(m\). We use SVVAMP 0.12.0 \cite{durand2016svvamp}, a Python package for studying the manipulability of voting rules. 
As for Plurality, the critical CM rate appears to converge to a limit in \((0,1)\) that increases with~\(m\).

\section{IRV (Instant-Runoff Voting)}\label{sec_irv}

Let us now proceed to \emph{Instant-Runoff Voting} (IRV), where the winner is determined through multiple rounds. In each round, the candidate with the lowest Plurality score is eliminated, until only one remains. Formally, let \( K(r, P) \) be the set of remaining candidates at the start of round \(r\) and \(\ell(r, P)\) be the candidate losing at round \(r\). We have:
\[
\left\{
\begin{array}{l}
	K(1, P) = \mathcal{C}(P), \\[2pt]
	\ell(r, P) = \arg \min s_{\Plu{}}(c, P_{K(r, P)}), \\[2pt]
	K(r+1, P) = K(r, P) \setminus \{ \ell(r, P) \},
\end{array}
\right.
\]
using a tie-breaking rule for elimination when necessary. The winner \(\IRV{}(P)\) is the last remaining candidate in \(K(m(P), P)\).

\subsection{Theoretical Results for IRV}

As usual, we start by examining the expected normalized profile \(\hat{P}\). Since \(\theta > 0\), candidate 1 clearly wins. Now suppose that IRV is CM in $\hat{P}$ to a target profile \( Q \), where candidate \(c \neq 1\) wins. Candidate~1 must be eliminated in some round \(r\). For conciseness, denote \(K = K(r, Q)\) and \(k = |K|\). Obviously \(c\) must belong to \(K\). The sincere voters' contribution to candidate 1’s score at this round is:
\[
s_\Plu{}(1, \hat{P}^{1 \succ c}_{K}) = s_\Plu{}(1, \hat{P}_{K}) = \theta + \frac{1-\theta}{k}.
\]
Thus, \(s_\Plu(1, Q_K) \geq \theta + \frac{1-\theta}{k}\), and since \(k \geq 2\), this is strictly greater than \(\frac{1}{k}\). Therefore, the score of candidate~1 exceeds the average score at this round, hence it cannot be minimal. This contradiction proves that IRV is not CM in $\hat{P}$.

In this reasoning, IRV’s resistance to coalitional manipulation stems from the fact that in any subset of candidates \(K\) containing candidate 1, this candidate has a Plurality score that exceeds the average score. This motivates the following definition:

\begin{definition}
	A candidate \(c\) is a \emph{Super Condorcet Winner} (SCW) in a profile \(P\) if, for every subset of candidates \(K\) containing \(c\), the following holds:
	\[
	s_{\Plu{}}(c, P_{K}) > \frac{w(P)}{|K|}.
	\]
\end{definition}

This concept strengthens the classical notion of a \emph{Condorcet Winner}, which only requires the condition to hold for subsets \(K\) of size 2. We summarize its relevance to IRV as follows:

\begin{lemma}\label{lem_scw_makes_irv_not_cm}
If \( c \) is an SCW in profile \( P \), then \(\IRV{}(P) = c\) and IRV is not CM in \( P \).
\end{lemma}

The same result easily extends to several IRV variants, such as Exhaustive Ballot \cite{durand2023coalitional}, Condorcet-IRV \cite{green2016statistical, durand2016condorcet}, Benham rule, Tideman rule, Smith-IRV, and Woodall rule \cite{green2011four}. However, the converse is not true: there exists profiles without an SCW where IRV is still not CM (see \citet[Table 1.1]{durand2015towards} for an example).

Now, consider the neighborhood of the expected profile \(\hat{P}\). Since the SCW condition involves a finite number of strict inequalities that depend continuously on the profile’s coefficients, candidate~1 is an SCW not only in \(\hat{P}\) but also in its neighborhood. From here, we can follow two proof strategies that only differ in the order of their steps.

One approach is to first apply Lemma~\ref{lem_scw_makes_irv_not_cm} to deduce that IRV is not CM in this neighborhood. Using Lemma~\ref{lem_not_CM_in_neighborhood} (based on the weak law of large numbers), we then deduce \(\lim_{n\to\infty} \rho(\IRV{}, m, n, \theta) = 0\). Alternatively, we could first use the weak law of large numbers to show that candidate 1 is an SCW with high probability, then apply Lemma~\ref{lem_scw_makes_irv_not_cm} to show \(\lim_{n\to\infty} \rho(\IRV{}, m, n, \theta) = 0\).

Since this holds for any $\theta > 0$, we obtain a remarkable result:

\begin{theorem}\label{thm_irv_theta_c}
For IRV, the critical value of the concentration parameter in Perturbed Culture is  
\[
\theta_c(\IRV{}, m) = 0.
\]	
\end{theorem}

In summary, within the Perturbed Culture model, IRV has the smallest possible critical value. Even a slight concentration of preferences favoring candidate~1 is enough for IRV to become resistant to coalitional manipulation with high probability. And for the same reasons, this also holds for the IRV variants mentioned earlier.

\subsection{Simulations for IRV}

\begin{figure}
	\centering
	\input{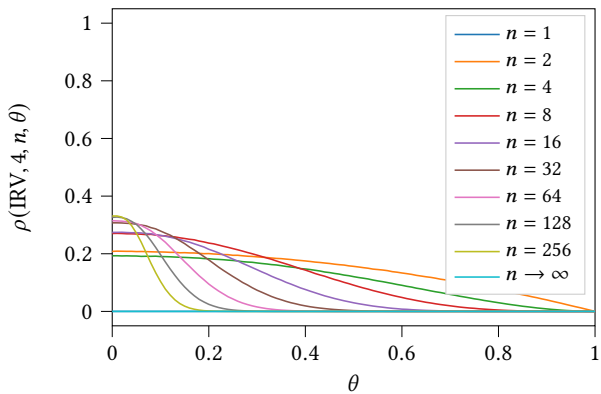}
	\caption{CM rate of IRV as a function of \( \theta \) for different values of \( n \) with \( m = 4 \). Curves for finite \( n \) are based on Monte Carlo simulations with 1,000,000 profiles per point. The limiting curve as \( n \to \infty \) follows from Theorem~\ref{thm_irv_theta_c}.}
	\label{fig_irv_cm_rate_of_theta_and_n_v}
	\Description{CM rate of IRV as a function of theta for different values of n with m = 4. Curves for finite n are based on Monte Carlo simulations with 1,000,000 profiles per point. The limiting curve as n tends to infinity follows from our theoretical results.}
\end{figure}

\begin{figure}
	\centering
\begin{tikzpicture}

\definecolor{crimson2143940}{RGB}{214,39,40}
\definecolor{darkgrey176}{RGB}{176,176,176}
\definecolor{darkorange25512714}{RGB}{255,127,14}
\definecolor{forestgreen4416044}{RGB}{44,160,44}
\definecolor{lightgrey204}{RGB}{204,204,204}
\definecolor{steelblue31119180}{RGB}{31,119,180}

\begin{axis}[reverse legend,
height=\axisHeight,
legend cell align={left},
legend style={font=\legendFont, fill opacity=1, draw opacity=1, text opacity=1, draw=lightgrey204},
legend pos=north west,
log basis x={10},
tick align=outside,
tick pos=left,
width=\axisWidth,
x grid style={darkgrey176},
xlabel={$n$},
xmin=1.0, xmax=1000,
xmode=log,
xtick style={color=black},
xmode=log,
y grid style={darkgrey176},
ylabel={$\rho(\text{IRV}, m, n, \theta_c(\text{IRV}, m))$},
ymin=-0.05, ymax=1.05,
ytick={-0.2, 0.0, 0.2, 0.4000000000000001, 0.6000000000000001, 0.8, 1.0000000000000002, 1.2000000000000002},
ytick style={color=black}
]
\addplot [semithick, steelblue31119180]
table {%
1 0
2 0.111114
3 0.11121
4 0.092343
5 0.093478
6 0.106596
7 0.123416
8 0.135911
9 0.135843
10 0.120496
11 0.120722
12 0.12596
13 0.13791
14 0.142326
15 0.142952
16 0.131458
17 0.131795
18 0.135156
19 0.14361
20 0.146303
21 0.146508
22 0.137987
23 0.138397
24 0.139302
25 0.147007
26 0.149324
27 0.149186
28 0.141765
29 0.142026
30 0.143381
31 0.149186
32 0.150832
33 0.151177
34 0.144543
35 0.145045
36 0.145094
37 0.151094
38 0.152413
39 0.15182
40 0.146951
41 0.147354
42 0.147086
43 0.152505
44 0.15327
45 0.153105
46 0.1483
47 0.148476
48 0.148265
49 0.15287
50 0.154506
51 0.154007
52 0.150058
53 0.149954
54 0.150089
55 0.153885
56 0.154942
57 0.153862
58 0.150842
59 0.151003
60 0.151155
61 0.155279
62 0.154559
63 0.155491
64 0.151916
65 0.151803
66 0.152379
67 0.155436
68 0.156133
69 0.156
70 0.152401
71 0.152234
72 0.153283
73 0.155697
74 0.156181
75 0.156163
76 0.153567
77 0.153091
78 0.153848
79 0.156394
80 0.156926
81 0.15654
82 0.154353
83 0.153978
84 0.154352
85 0.156944
86 0.156786
87 0.156687
88 0.155269
89 0.15503
90 0.155256
91 0.157988
92 0.157584
93 0.156904
94 0.15482
95 0.155039
96 0.155389
97 0.157911
98 0.15786
99 0.157651
100 0.155642
125 0.157198
158 0.15941
199 0.160735
251 0.161194
316 0.161807
398 0.162758
501 0.163715
630 0.163545
794 0.164782
1000 0.164906
};
\addlegendentry{$m=3$}
\addplot [semithick, darkorange25512714]
table {%
1 0
2 0.208057
3 0.213294
4 0.192875
5 0.202783
6 0.221184
7 0.248888
8 0.270623
9 0.271447
10 0.245401
11 0.248756
12 0.257542
13 0.280181
14 0.288546
15 0.289224
16 0.274593
17 0.277723
18 0.279433
19 0.293521
20 0.297455
21 0.297816
22 0.284377
23 0.284284
24 0.28849
25 0.30097
26 0.30321
27 0.305023
28 0.294298
29 0.296044
30 0.296547
31 0.306152
32 0.308109
33 0.306837
34 0.298153
35 0.298265
36 0.30075
37 0.309572
38 0.311314
39 0.311325
40 0.304207
41 0.306655
42 0.306097
43 0.312038
44 0.312836
45 0.313438
46 0.305603
47 0.30513
48 0.307594
49 0.314582
50 0.315544
51 0.314963
52 0.310399
53 0.311461
54 0.311265
55 0.316184
56 0.316801
57 0.317411
58 0.311622
59 0.310406
60 0.311774
61 0.316598
62 0.318061
63 0.31843
64 0.314023
65 0.315884
66 0.314426
67 0.318546
68 0.319293
69 0.320215
70 0.314839
71 0.314655
72 0.315218
73 0.319028
74 0.320283
75 0.32046
76 0.317919
77 0.31791
78 0.316977
79 0.320937
80 0.322363
81 0.322246
82 0.317127
83 0.317151
84 0.318066
85 0.321915
86 0.322908
87 0.322132
88 0.318976
89 0.320226
90 0.321043
91 0.323221
92 0.323793
93 0.323778
94 0.319759
95 0.318888
96 0.319951
97 0.324508
98 0.323025
99 0.324027
100 0.321106
125 0.324684
158 0.329101
199 0.330701
251 0.331086
316 0.333652
398 0.335096
501 0.336659
630 0.337643
794 0.339092
1000 0.339527
};
\addlegendentry{$m=4$}
\addplot [semithick, forestgreen4416044]
table {%
1 0
2 0.286119
3 0.297459
4 0.280046
5 0.298855
6 0.320359
7 0.356369
8 0.384422
9 0.388153
10 0.361115
11 0.365421
12 0.374873
13 0.401868
14 0.413206
15 0.414781
16 0.400223
17 0.405849
18 0.4069
19 0.4199
20 0.426239
21 0.427015
22 0.410618
23 0.412012
24 0.417272
25 0.432691
26 0.437823
27 0.439124
28 0.42884
29 0.432499
30 0.431399
31 0.441104
32 0.443061
33 0.444664
34 0.432769
35 0.432488
36 0.433556
37 0.444565
38 0.448005
39 0.446848
40 0.441206
41 0.444128
42 0.442242
43 0.450479
44 0.450564
45 0.451087
46 0.44303
47 0.443901
48 0.445374
49 0.453749
50 0.454166
51 0.455554
52 0.449767
53 0.451388
54 0.451104
55 0.45679
56 0.45665
57 0.457491
58 0.448581
59 0.448525
60 0.450448
61 0.458172
62 0.460002
63 0.459177
64 0.454733
65 0.457815
66 0.45583
67 0.460223
68 0.460673
69 0.460567
70 0.45579
71 0.455706
72 0.456224
73 0.461286
74 0.462071
75 0.462437
76 0.458961
77 0.461573
78 0.459085
79 0.46284
80 0.463544
81 0.462696
82 0.458144
83 0.458637
84 0.4597
85 0.464186
86 0.465402
87 0.464922
88 0.462133
89 0.463999
90 0.462061
91 0.465654
92 0.466448
93 0.466679
94 0.462385
95 0.462112
96 0.46111
97 0.466956
98 0.467666
99 0.466569
100 0.463835
125 0.469401
158 0.473359
199 0.477036
251 0.476525
316 0.479786
398 0.484022
501 0.484563
630 0.486411
794 0.487911
1000 0.488649
};
\addlegendentry{$m=5$}
\addplot [semithick, crimson2143940]
table {%
1 0
2 0.34938
3 0.365458
4 0.353967
5 0.378698
6 0.40339
7 0.444762
8 0.474923
9 0.483867
10 0.458118
11 0.464902
12 0.474213
13 0.504024
14 0.514887
15 0.518097
16 0.503728
17 0.512967
18 0.511113
19 0.526095
20 0.531411
21 0.533451
22 0.516556
23 0.520136
24 0.5251
25 0.541012
26 0.546242
27 0.547432
28 0.53796
29 0.542535
30 0.542656
31 0.552461
32 0.553341
33 0.554229
34 0.545355
35 0.543684
36 0.543984
37 0.555806
38 0.558548
39 0.558669
40 0.55305
41 0.556694
42 0.555897
43 0.562873
44 0.564173
45 0.563792
46 0.556405
47 0.55696
48 0.558297
49 0.567059
50 0.567397
51 0.568601
52 0.56328
53 0.565509
54 0.565134
55 0.569883
56 0.569742
57 0.569629
58 0.562193
59 0.56265
60 0.562402
61 0.570571
62 0.572667
63 0.572583
64 0.569226
65 0.571244
66 0.569959
67 0.575299
68 0.575203
69 0.574857
70 0.570053
71 0.570449
72 0.570851
73 0.576761
74 0.576562
75 0.577842
76 0.572945
77 0.575447
78 0.575234
79 0.577691
80 0.578109
81 0.577396
82 0.573578
83 0.572907
84 0.573933
85 0.57942
86 0.580869
87 0.580487
88 0.577137
89 0.578998
90 0.57693
91 0.581984
92 0.581488
93 0.581243
94 0.577662
95 0.577685
96 0.578056
97 0.581455
98 0.580792
99 0.582924
100 0.580187
125 0.586104
158 0.590754
199 0.593454
251 0.594431
316 0.599431
398 0.602193
501 0.603499
630 0.604748
794 0.607355
1000 0.609584
};
\addlegendentry{$m=6$}
\end{axis}

\end{tikzpicture}
	\caption{CM rate of IRV as a function of $n$ for different values of $m$ with $\theta = \theta_c(\IRV{}, m) = 0$ (Impartial Culture). Monte Carlo simulations with 1,000,000 profiles per point.}
	\label{fig_irv_cm_rate_at_theta_c_of_n_v_n_c}
	\Description{CM rate of IRV as a function of n for different values of m with critical theta, which is zero in that case (Impartial Culture). Monte Carlo simulations with 1,000,000 profiles per point.}
\end{figure}
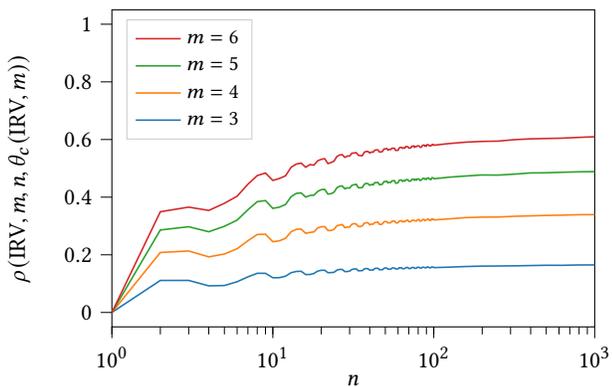

For IRV, as for Two-Round, our simulations for finite \(n\) are carried out using SVVAMP. Similar to Figure \ref{fig_plu_cm_rate_of_theta_and_n_v} for Plurality, Figure \ref{fig_irv_cm_rate_of_theta_and_n_v} shows the CM rate as a function of \(\theta\) for different \(n\). For large \(n\), the curve takes on a sigmoidal shape that converges to the theoretical curve from Theorem~\ref{thm_irv_theta_c}. The behavior near \(\theta = 0\) suggests that in the Impartial Culture model, the CM rate converges to a limit within \((0, 1)\), as proven for \(m=3\) by \citet{lepelley1999kimroush}, conjectured in the general case by \citet[Conjecture 7.8]{durand2015towards}, and further supported by Figure \ref{fig_irv_cm_rate_at_theta_c_of_n_v_n_c}. This second figure also indicates that, as with Plurality and Two-Round, the limit CM rate in the critical regime appears to increase with the number of candidates~\(m\).

\subsection{Empirical Results for IRV}

In the Perturbed Culture model, the presence of an SCW explains IRV's resistance to coalitional manipulation. However, as noted earlier, this is not a necessary condition: IRV can be non-manipulable in profiles without an SCW. This raises the question of whether the presence of an SCW often accounts for IRV's non-manipulability in realistic scenarios.

\begin{table}[t]
	\caption{Empirical study of Super Condorcet Winners (SCW) and IRV in two datasets, with the presence of a Condorcet Winner (CW) included as a reference.}
	\label{tab_empirical_scw_and_irv}
	\begin{tabular}{lrr}\toprule
		 & \textit{Netflix dataset} & \textit{FairVote dataset} \\ \midrule
		Profiles & 11,215 & 10,044 \\
		--- with a CW & 99.30\% & 99.98\% \\
		--- where IRV is not CM (a) & 95.87\% & 96.30\% \\
		--- with an SCW (b) & 94.05\% & 96.20\% \\
		Ratio (b) / (a) & 98\% & > 99\%${}^*$ \\
	\end{tabular}

	\bigskip\raggedright{\footnotesize${}^*$ We omit the next digit of the raw result (99.9\%), not significant given the sample size.} 
\end{table}

To investigate this, Table \ref{tab_empirical_scw_and_irv} analyzes the Netflix and FairVote datasets \cite{durand2023coalitional}, which respectively contain 11,215 profiles derived from slight perturbations of 2,243 empirical profiles and 10,044 profiles based on 162 empirical profiles.
It provides two key insights. First, an SCW is very common in real-world datasets, here appearing in 94\% or 96\% of profiles. Second, in most cases where IRV resists CM, this can be explained by the presence of an SCW---98\% in the Netflix dataset and over 99\% in the FairVote dataset. This confirms that SCWs are a crucial factor in IRV's resilience to manipulation.

The frequent appearance of SCWs may seem surprising, but it becomes intuitive when revisiting the definition. For a candidate~\(c\) and a subset \(K\) of candidates that includes \(c\), if preferences were perfectly balanced, we would expect \(s_\Plu{}(c, P_K) = \frac{w(P)}{|K|}\). The condition for being an SCW is simply to exceed this average. Therefore, even a slight bias in favor of \(c\) makes it likely for \(c\) to be an SCW.

\section{Convergence Speed}\label{sec_convergence_speed}

We will now study the convergence speed, to assess how fast the results found for $n \to \infty$ become relevant for finite values of $n$.

\subsection{Theoretical Bound}\label{sec_speed_theo}

We will first show that in the non-critical regime, the CM rate converges exponentially fast as \( n \to \infty \). Next, we will bound this speed of convergence depending on the parameter \(\theta\).

As an example, consider Lemma \ref{lem_not_CM_in_neighborhood}, where we assume the existence of a neighborhood where the rule \( f \) is not CM. By definition, there exists \(\epsilon > 0\) such that this neighborhood contains an open ball of diameter \(\epsilon\) for the infinity norm. Our approach is to apply Hoeffding's inequality \cite{hoeffding1994probability} to bound the probability that the normalized random profile \(\bar{P}\) falls outside this ball. Since Hoeffding's inequality applies to scalar random variables, we use the union bound to extend it to the weight vector representing a voting profile.
Formally, denoting by \(\mathbb{P}\) the probability:
\begin{align}
	\rho(f, m, n, \theta)
	&\leq \mathbb{P}( d(\bar{P}, \hat{P}) \geq \epsilon ), \nonumber\\
	&\leq \textstyle\sum_{p} \mathbb{P}( | w(p, \bar{P}) - w(p, \hat{P}) | \geq \epsilon ) \quad \text{(union bound)}, \nonumber\\
	&\leq 2 m! e^{- 2 \epsilon^2 n} \quad \text{(Hoeffding's inequality)}.\label{eq_bound_obtained_from_hoeffding}
\end{align}
Thus, the convergence is exponentially fast in \(n\), and we can quantify the rate if the size of the neighborhood is known. By the same reasoning, similar results hold for Lemmas \ref{lem_um_in_neighborhood} and \ref{lem_eta_stable_cm}, and consequently for Theorems \ref{thm_plurality_theta_c}, \ref{thm_tr_theta_c}, and \ref{thm_irv_theta_c}. In the literature on phase transitions, this is known as \emph{sharp} transitions, meaning that the limiting curve quickly approximates the behavior even for finite \(n\).

Let us now bound the speed of convergence more precisely. For example, consider Plurality in the supercritical regime. If a profile~\(P\) lies within an open ball of radius \(\epsilon\) centered at \(\hat{P}\), the score of candidate 1 is bounded from below: \( s_\Plu{}(1, P) > \theta + \frac{1 - \theta}{m} - m! \epsilon \), and the number of manipulators is bounded from above: \( w(P^{2 \succ 1}) < \frac{1 - \theta}{2} + m! \epsilon \). To ensure \( s_\Plu{}(1, P) > w(P^{2 \succ 1}) \), we set \(\epsilon = \frac{(3 m - 2) \theta - (m - 2)}{2 m!}\), which can be rewritten as \(\epsilon = \big(\theta - \theta_c(\Plu{}, m)\big) \frac{3 m - 2}{2 m!}\). Using our bound \eqref{eq_bound_obtained_from_hoeffding}, there exists a coefficient \(A^+(\Plu{}, m)\)---which we could explicitly compute---such that:
\[
\rho(\Plu{}, m, n, \theta) = O\left(e^{- A^+(\Plu{}, m) (\theta - \theta_c(\Plu{}, m))^2}\right).
\]

This reasoning for Plurality in the supercritical regime generalizes to the subcritical regime, with a coefficient \(A^-(\Plu{}, m)\), and to the other voting rules in this paper: since all relevant quantities (scores, numbers of manipulators) are linear in the profile weights, we can take a value of \(\epsilon\) that depends linearly on \(\theta - \theta_c\), leading to a term in $(\theta - \theta_c)^2$ via Hoeffding’s inequality. Thus, we obtain:
\begin{itemize}
	\item Supercritical regime: \(\rho = O\big(e^{-A^+(f, m) (\theta - \theta_c)^2 n}\big)\),
	\item Subcritical regime: \(\rho = 1 - O\big(e^{-A^-(f, m) (\theta_c - \theta)^2 n}\big)\),
\end{itemize}
where \(\rho = \rho(f, m, n, \theta)\) and \(\theta_c = \theta_c(f, m)\).

\subsection{Simulation Study of the Convergence Speed}

\begin{figure}[!t]
	\centering
\begin{tikzpicture}

\definecolor{crimson2143940}{RGB}{214,39,40}
\definecolor{darkgrey176}{RGB}{176,176,176}
\definecolor{darkorange25512714}{RGB}{255,127,14}
\definecolor{darkturquoise23190207}{RGB}{23,190,207}
\definecolor{forestgreen4416044}{RGB}{44,160,44}
\definecolor{goldenrod18818934}{RGB}{188,189,34}
\definecolor{grey127}{RGB}{127,127,127}
\definecolor{lightgrey204}{RGB}{204,204,204}
\definecolor{mediumpurple148103189}{RGB}{148,103,189}
\definecolor{orchid227119194}{RGB}{227,119,194}
\definecolor{sienna1408675}{RGB}{140,86,75}
\definecolor{steelblue31119180}{RGB}{31,119,180}

\begin{axis}[
height=\axisHeight,
legend cell align={left},
legend style={font=\legendFont, 
  fill opacity=1,
  draw opacity=1,
  text opacity=1,
  at={(0.97,0.03)},
  anchor=south east,
  draw=lightgrey204
},
log basis y={10},
tick align=outside,
tick pos=left,
width=\axisWidth,
x grid style={darkgrey176},
xlabel={$n$},
xmin=-49.1, xmax=1075.1,
xtick style={color=black},
y grid style={darkgrey176},
ylabel={$\rho(\text{Plu}, 4, n, \theta)$},
ymin=0.001, ymax=1,
ymode=log,
ymode=log
]
\addplot [semithick, steelblue31119180]
table {%
2 0.131701
3 0.266031
4 0.27462
5 0.362872
6 0.376014
7 0.419457
8 0.438869
10 0.478969
11 0.492667
13 0.516011
16 0.538624
19 0.554078
23 0.566968
27 0.574274
32 0.579256
38 0.583498
45 0.586284
54 0.587797
64 0.589441
76 0.588819
91 0.5893
108 0.588458
128 0.588256
152 0.587324
181 0.585298
215 0.583092
256 0.580998
304 0.577356
362 0.573796
431 0.570351
512 0.565967
609 0.561397
724 0.555942
861 0.549114
1024 0.542335
};
\addlegendentry{$\theta = 0.205$}
\addplot [semithick, darkorange25512714]
table {%
2 0.130418
3 0.263371
4 0.272599
5 0.362799
6 0.372854
7 0.417658
8 0.435832
10 0.474556
11 0.488771
13 0.509698
16 0.532097
19 0.547084
23 0.559076
27 0.565467
32 0.569994
38 0.573002
45 0.575979
54 0.574659
64 0.57471
76 0.574206
91 0.572571
108 0.570327
128 0.568874
152 0.566056
181 0.563114
215 0.558122
256 0.552644
304 0.547537
362 0.541585
431 0.534873
512 0.526349
609 0.518822
724 0.508606
861 0.498669
1024 0.487185
};
\addlegendentry{$\theta = 0.208$}
\addplot [semithick, forestgreen4416044]
table {%
2 0.129812
3 0.263442
4 0.271005
5 0.360371
6 0.370705
7 0.414022
8 0.430613
10 0.4694
11 0.482612
13 0.503237
16 0.524217
19 0.537719
23 0.548307
27 0.553741
32 0.557389
38 0.559037
45 0.559214
54 0.558466
64 0.555958
76 0.553984
91 0.550347
108 0.54687
128 0.541516
152 0.537118
181 0.531443
215 0.523957
256 0.515814
304 0.507013
362 0.496766
431 0.486919
512 0.474432
609 0.461423
724 0.446974
861 0.431321
1024 0.414239
};
\addlegendentry{$\theta = 0.212$}
\addplot [semithick, crimson2143940]
table {%
2 0.12674
3 0.260536
4 0.266004
5 0.354717
6 0.362208
7 0.405543
8 0.422232
10 0.457874
11 0.469916
13 0.488742
16 0.50683
19 0.518321
23 0.526389
27 0.529547
32 0.530177
38 0.530326
45 0.527985
54 0.523548
64 0.518466
76 0.513757
91 0.506478
108 0.498291
128 0.490217
152 0.479468
181 0.468381
215 0.455847
256 0.442581
304 0.427126
362 0.409835
431 0.39146
512 0.371911
609 0.350709
724 0.328211
861 0.304744
1024 0.277009
};
\addlegendentry{$\theta = 0.220$}
\addplot [semithick, mediumpurple148103189]
table {%
2 0.123099
3 0.256436
4 0.259058
5 0.347703
6 0.352206
7 0.395194
8 0.407716
10 0.441649
11 0.452656
13 0.469375
16 0.483858
19 0.492745
23 0.496859
27 0.496414
32 0.494831
38 0.490839
45 0.484734
54 0.477031
64 0.468211
76 0.457868
91 0.44556
108 0.432609
128 0.417598
152 0.401881
181 0.385377
215 0.365718
256 0.344698
304 0.323199
362 0.298701
431 0.272855
512 0.246154
609 0.219489
724 0.191708
861 0.163952
1024 0.136786
};
\addlegendentry{$\theta = 0.231$}
\addplot [semithick, sienna1408675]
table {%
2 0.117343
3 0.251495
4 0.247903
5 0.33484
6 0.335673
7 0.37647
8 0.386446
10 0.414926
11 0.424
13 0.436973
16 0.446076
19 0.44931
23 0.447947
27 0.444304
32 0.436906
38 0.427104
45 0.416178
54 0.401161
64 0.387059
76 0.369452
91 0.349445
108 0.330027
128 0.308724
152 0.284739
181 0.261126
215 0.234512
256 0.207565
304 0.180434
362 0.15373
431 0.127821
512 0.102502
609 0.079175
724 0.058407
861 0.041131
1024 0.027312
};
\addlegendentry{$\theta = 0.249$}
\addplot [semithick, orchid227119194]
table {%
2 0.108881
3 0.240866
4 0.230356
5 0.315647
6 0.310016
7 0.348764
8 0.351482
10 0.372689
11 0.379802
13 0.385305
16 0.387061
19 0.382971
23 0.374408
27 0.363238
32 0.348538
38 0.33287
45 0.314633
54 0.292661
64 0.270027
76 0.245923
91 0.221544
108 0.195472
128 0.16969
152 0.14377
181 0.118662
215 0.094583
256 0.072527
304 0.053788
362 0.037249
431 0.024553
512 0.014753
609 0.008286
724 0.004204
861 0.00176
1024 0.000722
};
\addlegendentry{$\theta = 0.277$}
\addplot [semithick, grey127]
table {%
2 0.095338
3 0.223213
4 0.202745
5 0.28375
6 0.267497
7 0.302349
8 0.297343
10 0.307732
11 0.310316
13 0.3066
16 0.29819
19 0.284955
23 0.267288
27 0.248959
32 0.227717
38 0.205164
45 0.181548
54 0.15543
64 0.131722
76 0.108502
91 0.084555
108 0.065189
128 0.047784
152 0.033228
181 0.021326
215 0.012764
256 0.006881
304 0.003487
362 0.00142
431 0.000496
512 0.00018
609 5e-05
724 8e-06
};
\addlegendentry{$\theta = 0.321$}
\addplot [semithick, goldenrod18818934]
table {%
2 0.07724
3 0.195688
4 0.16043
5 0.231137
6 0.202854
7 0.231288
8 0.215228
10 0.213135
11 0.210816
13 0.198426
16 0.177949
19 0.158912
23 0.135532
27 0.115865
32 0.095551
38 0.074832
45 0.057467
54 0.040325
64 0.02801
76 0.017822
91 0.010362
108 0.005481
128 0.002675
152 0.001156
181 0.000386
215 0.000127
256 3.2e-05
304 6e-06
};
\addlegendentry{$\theta = 0.390$}
\addplot [semithick, darkturquoise23190207]
table {%
2 0.051992
3 0.147091
4 0.100519
5 0.151947
6 0.113939
7 0.132778
8 0.109474
10 0.097233
11 0.092703
13 0.07722
16 0.057985
19 0.043688
23 0.030015
27 0.020262
32 0.01288
38 0.007337
45 0.00369
54 0.001613
64 0.00067
76 0.000206
91 5.5e-05
108 7e-06
128 3e-06
152 1e-06
};
\addlegendentry{$\theta = 0.500$}
\end{axis}

\end{tikzpicture}
	\caption{CM rate of Plurality as a function of $n$ for different supercritical values of $\theta$ with $m = 4$. Monte Carlo simulations with 1,000,000 profiles per point.}
	\label{fig_plu_cm_rate_of_n_v_and_theta_supercritical}
	\Description{CM rate of Plurality as a function of n for different supercritical values of theta with m = 4. Monte Carlo simulations with 1,000,000 profiles per point.}
\end{figure}
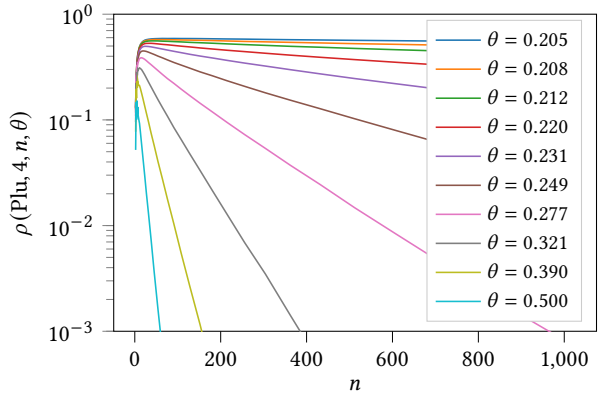

Figure \ref{fig_plu_cm_rate_of_n_v_and_theta_supercritical} shows the CM rate of Plurality as a function of \(n\) for various values of \(\theta\) (whereas Figure \ref{fig_plu_cm_rate_of_theta_and_n_v} does the reverse). Each curve has an oblique asymptote on a semi-log scale, indicating not only that it is bounded by a decreasing exponential (as predicted by theory), but that it follows the form \(\rho \sim_{n \to \infty} B(m, \theta) e^{- C(m, \theta) n}\), with \(m = 4\) here. This figure also allows us to measure the slopes of the asymptotes, providing the values of \(C(4, \theta)\) for each \(\theta\).

\begin{figure}[!t]
	\centering
\begin{tikzpicture}

\definecolor{darkgrey176}{RGB}{176,176,176}
\definecolor{darkorange25512714}{RGB}{255,127,14}
\definecolor{forestgreen4416044}{RGB}{44,160,44}
\definecolor{lightgrey204}{RGB}{204,204,204}
\definecolor{steelblue31119180}{RGB}{31,119,180}

\begin{axis}[reverse legend,
height=\axisHeight,
legend cell align={left},
legend style={font=\legendFont, 
  fill opacity=1,
  draw opacity=1,
  text opacity=1,
  at={(0.03,0.97)},
  anchor=north west,
  draw=lightgrey204
},
log basis x={10},
log basis y={10},
tick align=outside,
tick pos=left,
width=\axisWidth,
x grid style={darkgrey176},
xtick={0.005, 0.020, 0.077, 0.300},
xticklabels={0.005, 0.020, 0.077, 0.300},
xlabel={$\theta - \theta_c(\text{Plu}, 4)$},
xmin=0.00407438854172251, xmax=0.368153401336105,
xmode=log,
xtick style={color=black},
xmode=log,
y grid style={darkgrey176},
ylabel={Asymptotic negative slope $C(4, \theta)$},
ymin=1.68890102412472e-05, ymax=0.171212070097638,
ymode=log,
ymode=log
]
\path [draw=steelblue31119180, semithick]
(axis cs:0.00499999999999998,7.10675e-05)
--(axis cs:0.00499999999999998,9.55415e-05);

\path [draw=steelblue31119180, semithick]
(axis cs:0.00799999999999998,0.000127902)
--(axis cs:0.00799999999999998,0.000174132);

\path [draw=steelblue31119180, semithick]
(axis cs:0.012,0.000232843)
--(axis cs:0.012,0.000297139);

\path [draw=steelblue31119180, semithick]
(axis cs:0.02,0.000529299)
--(axis cs:0.02,0.000621501);

\path [draw=steelblue31119180, semithick]
(axis cs:0.031,0.001064335)
--(axis cs:0.031,0.001230741);

\path [draw=steelblue31119180, semithick]
(axis cs:0.049,0.002453784)
--(axis cs:0.049,0.00271245);

\path [draw=steelblue31119180, semithick]
(axis cs:0.077,0.005118355)
--(axis cs:0.077,0.007139879);

\path [draw=steelblue31119180, semithick]
(axis cs:0.121,0.011569868)
--(axis cs:0.121,0.018366818);

\path [draw=steelblue31119180, semithick]
(axis cs:0.19,0.022903157)
--(axis cs:0.19,0.049082839);

\path [draw=steelblue31119180, semithick]
(axis cs:0.3,0.074879445)
--(axis cs:0.3,0.112576215);

\addplot [semithick, darkorange25512714]
table {%
0.00499999999999998 8.35037478662248e-05
0.00799999999999998 0.000151327461274696
0.012 0.00025273953897666
0.02 0.000482293997918368
0.031 0.000839617448725901
0.049 0.00149832820172752
0.077 0.00265411465527282
0.121 0.00470145632659927
0.19 0.00832017034304093
0.3 0.0148274783959487
};
\addlegendentry{$y = 0.068 x^{1.265}$}
\addplot [semithick, forestgreen4416044]
table {%
0.00499999999999998 2.56857312648515e-05
0.00799999999999998 6.57245739891858e-05
0.012 0.00014782034333156
0.02 0.000410402367209771
0.031 0.000985559666177284
0.049 0.00246123373864955
0.077 0.00607499402911179
0.121 0.0149947369379039
0.19 0.036955522255088
0.3 0.0920908081101323
};
\addlegendentry{$y = 1.022 x^{1.999}$}
\addplot [semithick, steelblue31119180, mark=*, mark size=1, mark options={solid}, only marks]
table {%
0.00499999999999998 8.33045e-05
0.00799999999999998 0.000151017
0.012 0.000264991
0.02 0.0005754
0.031 0.001147538
0.049 0.002583117
0.077 0.006129117
0.121 0.014968343
0.19 0.035992998
0.3 0.09372783
};
\addlegendentry{Measured slope}
\end{axis}

\end{tikzpicture}
	\caption{Asymptotic negative slope \(C(4, \theta)\) from Figure~\ref{fig_plu_cm_rate_of_n_v_and_theta_supercritical}, plotted as a function of \(\theta - \theta_c(\Plu{}, 4)\). The vertical blue lines represent error margins.\protect\footnotemark}
	\label{fig_slopes}
	\Description{Asymptotic negative slope C(4, theta) measured in previous figure as a function of theta minus the critical threshold. The vertical blue lines represent error margins.}
\end{figure}
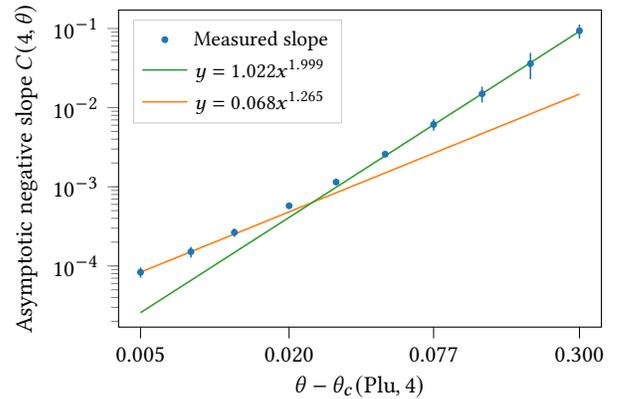
\footnotetext{The computation of the error margins is detailed in the code repository \url{https://github.com/francois-durand/irv-cm-aamas-2025}.}

To analyze how convergence speed varies with \(\theta\), Figure \ref{fig_slopes} plots the measured asymptotic slopes against \(\theta - \theta_c\) in log-log scale (the values of $\theta$ were specifically chosen to be evenly spaced in that figure). When $\theta$ is far from $\theta_c$, the dependency is in $(\theta - \theta_c)^2$, in line with the upper bound found previously. However, close to $\theta_c$, the dependency seems to involve a smaller exponent (estimated at $1.265$). In the terminology of phase transition, this is called the \emph{critical exponent} of the convergence speed.

We repeated this for \(m \in \{5, 6, 7\}\), the subcritical regime, and the other voting rules, with similar results but various critical exponents. 
This suggests a long-range dependency in $|\theta - \theta_c|^2$ but smaller critical exponents near the critical regime. This intriguing behavior will deserve further theoretical investigation.

\section{Future Work}\label{sec_conclusion}

A natural direction for future work is to compute the critical parameter \(\theta_c\) for other voting systems. Another key area of research would be a deeper analysis of the critical regime, including the calculation of the limiting CM rate at \(\theta = \theta_c\) and the asymptotic behavior of the slope of the sigmoid $\rho(\theta)$ at \( \theta = \theta_c \), which is linked to a finer analysis of the convergence speed in the non-critical regime. Expanding the study to other models, such as Mallows, is also promising. Preliminary analysis shows that the qualitative results observed in this paper, particularly the key finding that IRV's limit CM rate drops to zero with even slight concentration of preferences, also hold true under the Mallows model.




\begin{acks}
This work has been carried out at LINCS (\url{www.lincs.fr}).
\end{acks}

\newpage

\balance



\bibliographystyle{ACM-Reference-Format} 
\bibliography{Why_IRV_is_so_resilient_to_coalitional_manipulation}


\begin{thebibliography}{34}


\ifx \showCODEN    \undefined \def \showCODEN     #1{\unskip}     \fi
\ifx \showDOI      \undefined \def \showDOI       #1{#1}\fi
\ifx \showISBNx    \undefined \def \showISBNx     #1{\unskip}     \fi
\ifx \showISBNxiii \undefined \def \showISBNxiii  #1{\unskip}     \fi
\ifx \showISSN     \undefined \def \showISSN      #1{\unskip}     \fi
\ifx \showLCCN     \undefined \def \showLCCN      #1{\unskip}     \fi
\ifx \shownote     \undefined \def \shownote      #1{#1}          \fi
\ifx \showarticletitle \undefined \def \showarticletitle #1{#1}   \fi
\ifx \showURL      \undefined \def \showURL       {\relax}        \fi
\providecommand\bibfield[2]{#2}
\providecommand\bibinfo[2]{#2}
\providecommand\natexlab[1]{#1}
\providecommand\showeprint[2][]{arXiv:#2}

\bibitem[\protect\citeauthoryear{Bollob{\'a}s and Riordan}{Bollob{\'a}s and
  Riordan}{2006}]%
        {bollobas2006percolation}
\bibfield{author}{\bibinfo{person}{B{\'e}la Bollob{\'a}s} {and}
  \bibinfo{person}{Oliver Riordan}.} \bibinfo{year}{2006}\natexlab{}.
\newblock \bibinfo{booktitle}{\emph{Percolation}}.
\newblock \bibinfo{publisher}{Cambridge University Press}.
\newblock


\bibitem[\protect\citeauthoryear{Brandt, Conitzer, Endriss, Lang, and
  Procaccia}{Brandt et~al\mbox{.}}{2016}]%
        {brandt2016handbook}
\bibfield{author}{\bibinfo{person}{Felix Brandt}, \bibinfo{person}{Vincent
  Conitzer}, \bibinfo{person}{Ulle Endriss}, \bibinfo{person}{J{\'e}r{\^o}me
  Lang}, {and} \bibinfo{person}{Ariel Procaccia}.}
  \bibinfo{year}{2016}\natexlab{}.
\newblock \bibinfo{booktitle}{\emph{Handbook of computational social choice}}.
\newblock \bibinfo{publisher}{Cambridge University Press}.
\newblock


\bibitem[\protect\citeauthoryear{Chamberlin, Cohen, and Coombs}{Chamberlin
  et~al\mbox{.}}{1984}]%
        {chamberlin1984observed}
\bibfield{author}{\bibinfo{person}{John Chamberlin}, \bibinfo{person}{Jerry
  Cohen}, {and} \bibinfo{person}{Clyde Coombs}.}
  \bibinfo{year}{1984}\natexlab{}.
\newblock \showarticletitle{Social choice observed: Five presidential elections
  of the American Psychological Association}.
\newblock \bibinfo{journal}{\emph{The Journal of Politics}}
  \bibinfo{volume}{46}, \bibinfo{number}{2} (\bibinfo{year}{1984}),
  \bibinfo{pages}{479--502}.
\newblock


\bibitem[\protect\citeauthoryear{Chikazumi}{Chikazumi}{1997}]%
        {chikazumi1997physics}
\bibfield{author}{\bibinfo{person}{S{\'o}shin Chikazumi}.}
  \bibinfo{year}{1997}\natexlab{}.
\newblock \bibinfo{booktitle}{\emph{Physics of ferromagnetism}}.
\newblock \bibinfo{publisher}{Oxford University Press}.
\newblock


\bibitem[\protect\citeauthoryear{Christensen}{Christensen}{2002}]%
        {christensen2002percolation}
\bibfield{author}{\bibinfo{person}{Kim Christensen}.}
  \bibinfo{year}{2002}\natexlab{}.
\newblock \showarticletitle{Percolation theory}.
\newblock \bibinfo{journal}{\emph{Imperial College London}}
  \bibinfo{volume}{1} (\bibinfo{year}{2002}), \bibinfo{pages}{87}.
\newblock


\bibitem[\protect\citeauthoryear{Duminil-Copin}{Duminil-Copin}{2017}]%
        {duminil2017lectures}
\bibfield{author}{\bibinfo{person}{Hugo Duminil-Copin}.}
  \bibinfo{year}{2017}\natexlab{}.
\newblock \showarticletitle{Lectures on the Ising and Potts models on the
  hypercubic lattice}.
\newblock In \bibinfo{booktitle}{\emph{PIMS-CRM Summer School in Probability}}.
  \bibinfo{publisher}{Springer}, \bibinfo{pages}{35--161}.
\newblock


\bibitem[\protect\citeauthoryear{Duminil-Copin}{Duminil-Copin}{2018}]%
        {duminil2018sixty}
\bibfield{author}{\bibinfo{person}{Hugo Duminil-Copin}.}
  \bibinfo{year}{2018}\natexlab{}.
\newblock \showarticletitle{Sixty years of percolation}. In
  \bibinfo{booktitle}{\emph{Proceedings of the International Congress of
  Mathematicians: Rio de Janeiro 2018}}. \bibinfo{pages}{2829--2856}.
\newblock


\bibitem[\protect\citeauthoryear{Durand}{Durand}{2015}]%
        {durand2015towards}
\bibfield{author}{\bibinfo{person}{Fran{\c{c}}ois Durand}.}
  \bibinfo{year}{2015}\natexlab{}.
\newblock \emph{\bibinfo{title}{Towards less manipulable voting systems}}.
\newblock \bibinfo{thesistype}{Ph.D. Dissertation}.
  \bibinfo{school}{Universit{\'e} Pierre et Marie Curie-Paris VI}.
\newblock


\bibitem[\protect\citeauthoryear{Durand}{Durand}{2023}]%
        {durand2023coalitional}
\bibfield{author}{\bibinfo{person}{Fran{\c{c}}ois Durand}.}
  \bibinfo{year}{2023}\natexlab{}.
\newblock \showarticletitle{Coalitional manipulation of voting rules:
  Simulations on empirical data}.
\newblock \bibinfo{journal}{\emph{Constitutional Political Economy}}
  \bibinfo{volume}{34}, \bibinfo{number}{3} (\bibinfo{year}{2023}),
  \bibinfo{pages}{390--409}.
\newblock


\bibitem[\protect\citeauthoryear{Durand, Mathieu, and Noirie}{Durand
  et~al\mbox{.}}{2016a}]%
        {durand2016condorcet}
\bibfield{author}{\bibinfo{person}{Fran{\c c}ois Durand},
  \bibinfo{person}{Fabien Mathieu}, {and} \bibinfo{person}{Ludovic Noirie}.}
  \bibinfo{year}{2016}\natexlab{a}.
\newblock \showarticletitle{Can a {C}ondorcet rule have a low coalitional
  manipulability?}. In \bibinfo{booktitle}{\emph{European Conference on
  Artificial Intelligence (ECAI)}}, Vol.~\bibinfo{volume}{285}.
  \bibinfo{pages}{707--715}.
\newblock


\bibitem[\protect\citeauthoryear{Durand, Mathieu, and Noirie}{Durand
  et~al\mbox{.}}{2016b}]%
        {durand2016svvamp}
\bibfield{author}{\bibinfo{person}{Fran{\c{c}}ois Durand},
  \bibinfo{person}{Fabien Mathieu}, {and} \bibinfo{person}{Ludovic Noirie}.}
  \bibinfo{year}{2016}\natexlab{b}.
\newblock \showarticletitle{SVVAMP: {S}imulator of various voting algorithms in
  manipulating populations}. In \bibinfo{booktitle}{\emph{Proceedings of the
  AAAI conference on artificial intelligence}}, Vol.~\bibinfo{volume}{30}.
\newblock


\bibitem[\protect\citeauthoryear{Eggers and Vivyan}{Eggers and Vivyan}{2020}]%
        {eggers2020votes}
\bibfield{author}{\bibinfo{person}{Andrew Eggers} {and} \bibinfo{person}{Nick
  Vivyan}.} \bibinfo{year}{2020}\natexlab{}.
\newblock \showarticletitle{Who votes more strategically?}
\newblock \bibinfo{journal}{\emph{American Political Science Review}}
  \bibinfo{volume}{114}, \bibinfo{number}{2} (\bibinfo{year}{2020}),
  \bibinfo{pages}{470--485}.
\newblock


\bibitem[\protect\citeauthoryear{Favardin, Lepelley, and Serais}{Favardin
  et~al\mbox{.}}{2002}]%
        {favardin2002bordacopeland}
\bibfield{author}{\bibinfo{person}{Pierre Favardin}, \bibinfo{person}{Dominique
  Lepelley}, {and} \bibinfo{person}{Jérôme Serais}.}
  \bibinfo{year}{2002}\natexlab{}.
\newblock \showarticletitle{Borda rule, {C}opeland method and strategic
  manipulation}.
\newblock \bibinfo{journal}{\emph{Review of Economic Design}}
  \bibinfo{volume}{7} (\bibinfo{year}{2002}), \bibinfo{pages}{213--228}.
\newblock
Issue 2.


\bibitem[\protect\citeauthoryear{Gehrlein}{Gehrlein}{2006}]%
        {gehrlein2006condorcet}
\bibfield{author}{\bibinfo{person}{William Gehrlein}.}
  \bibinfo{year}{2006}\natexlab{}.
\newblock \bibinfo{booktitle}{\emph{Condorcet's Paradox}}.
\newblock \bibinfo{publisher}{Springer}.
\newblock


\bibitem[\protect\citeauthoryear{Gibbard}{Gibbard}{1973}]%
        {gibbard1973manipulation}
\bibfield{author}{\bibinfo{person}{Allan Gibbard}.}
  \bibinfo{year}{1973}\natexlab{}.
\newblock \showarticletitle{Manipulation of voting schemes: A general result}.
\newblock \bibinfo{journal}{\emph{Econometrica}} \bibinfo{volume}{41},
  \bibinfo{number}{4} (\bibinfo{year}{1973}), \bibinfo{pages}{587--601}.
\newblock


\bibitem[\protect\citeauthoryear{Green-Armytage}{Green-Armytage}{2011}]%
        {green2011four}
\bibfield{author}{\bibinfo{person}{James Green-Armytage}.}
  \bibinfo{year}{2011}\natexlab{}.
\newblock \showarticletitle{Four {C}ondorcet-{H}are hybrid methods for
  single-winner elections}.
\newblock \bibinfo{journal}{\emph{Voting matters}} \bibinfo{volume}{29},
  \bibinfo{number}{1} (\bibinfo{year}{2011}), \bibinfo{pages}{1--14}.
\newblock


\bibitem[\protect\citeauthoryear{Green-Armytage}{Green-Armytage}{2014}]%
        {green2014strategic}
\bibfield{author}{\bibinfo{person}{James Green-Armytage}.}
  \bibinfo{year}{2014}\natexlab{}.
\newblock \showarticletitle{Strategic voting and nomination}.
\newblock \bibinfo{journal}{\emph{Social Choice and Welfare}}
  \bibinfo{volume}{42}, \bibinfo{number}{1} (\bibinfo{year}{2014}),
  \bibinfo{pages}{111--138}.
\newblock


\bibitem[\protect\citeauthoryear{Green-Armytage, Tideman, and
  Cosman}{Green-Armytage et~al\mbox{.}}{2016}]%
        {green2016statistical}
\bibfield{author}{\bibinfo{person}{James Green-Armytage},
  \bibinfo{person}{Nicolaus Tideman}, {and} \bibinfo{person}{Rafael Cosman}.}
  \bibinfo{year}{2016}\natexlab{}.
\newblock \showarticletitle{Statistical evaluation of voting rules}.
\newblock \bibinfo{journal}{\emph{Social Choice and Welfare}}
  \bibinfo{volume}{46} (\bibinfo{year}{2016}), \bibinfo{pages}{183--212}.
\newblock


\bibitem[\protect\citeauthoryear{Hoeffding}{Hoeffding}{1994}]%
        {hoeffding1994probability}
\bibfield{author}{\bibinfo{person}{Wassily Hoeffding}.}
  \bibinfo{year}{1994}\natexlab{}.
\newblock \showarticletitle{Probability inequalities for sums of bounded random
  variables}.
\newblock \bibinfo{journal}{\emph{The collected works of Wassily Hoeffding}}
  (\bibinfo{year}{1994}), \bibinfo{pages}{409--426}.
\newblock


\bibitem[\protect\citeauthoryear{Kadanoff}{Kadanoff}{2000}]%
        {kadanoff2000statistical}
\bibfield{author}{\bibinfo{person}{Leo Kadanoff}.}
  \bibinfo{year}{2000}\natexlab{}.
\newblock \bibinfo{booktitle}{\emph{Statistical physics: Statics, dynamics and
  renormalization}}.
\newblock \bibinfo{publisher}{World Scientific}.
\newblock


\bibitem[\protect\citeauthoryear{Kim and Roush}{Kim and Roush}{1996}]%
        {kim1996manipulability}
\bibfield{author}{\bibinfo{person}{K.H. Kim} {and} \bibinfo{person}{F.W.
  Roush}.} \bibinfo{year}{1996}\natexlab{}.
\newblock \showarticletitle{Statistical manipulability of social choice
  functions}.
\newblock \bibinfo{journal}{\emph{Group Decision and Negotiation}}
  \bibinfo{volume}{5} (\bibinfo{year}{1996}), \bibinfo{pages}{263--282}.
\newblock
Issue 3.


\bibitem[\protect\citeauthoryear{Komargodski and Simmons-Duffin}{Komargodski
  and Simmons-Duffin}{2017}]%
        {komargodski2017random}
\bibfield{author}{\bibinfo{person}{Zohar Komargodski} {and}
  \bibinfo{person}{David Simmons-Duffin}.} \bibinfo{year}{2017}\natexlab{}.
\newblock \showarticletitle{The random-bond Ising model in 2.01 and 3
  dimensions}.
\newblock \bibinfo{journal}{\emph{Journal of Physics A: Mathematical and
  Theoretical}} \bibinfo{volume}{50}, \bibinfo{number}{15}
  (\bibinfo{year}{2017}), \bibinfo{pages}{154001}.
\newblock


\bibitem[\protect\citeauthoryear{Kos, Poland, Simmons-Duffin, and Vichi}{Kos
  et~al\mbox{.}}{2016}]%
        {kos2016precision}
\bibfield{author}{\bibinfo{person}{Filip Kos}, \bibinfo{person}{David Poland},
  \bibinfo{person}{David Simmons-Duffin}, {and} \bibinfo{person}{Alessandro
  Vichi}.} \bibinfo{year}{2016}\natexlab{}.
\newblock \showarticletitle{Precision islands in the Ising and O(N) models}.
\newblock \bibinfo{journal}{\emph{Journal of High Energy Physics}}
  \bibinfo{volume}{2016}, \bibinfo{number}{8} (\bibinfo{year}{2016}),
  \bibinfo{pages}{1--16}.
\newblock


\bibitem[\protect\citeauthoryear{Lepelley and Mbih}{Lepelley and Mbih}{1994}]%
        {lepelley1994vulnerability}
\bibfield{author}{\bibinfo{person}{Dominique Lepelley} {and}
  \bibinfo{person}{Boniface Mbih}.} \bibinfo{year}{1994}\natexlab{}.
\newblock \showarticletitle{The vulnerability of four social choice functions
  to coalitional manipulation of preferences}.
\newblock \bibinfo{journal}{\emph{Social Choice and Welfare}}
  \bibinfo{volume}{11} (\bibinfo{year}{1994}), \bibinfo{pages}{253--265}.
\newblock
Issue 3.


\bibitem[\protect\citeauthoryear{Lepelley and Valognes}{Lepelley and
  Valognes}{1999}]%
        {lepelley1999kimroush}
\bibfield{author}{\bibinfo{person}{Dominique Lepelley} {and}
  \bibinfo{person}{Fabrice Valognes}.} \bibinfo{year}{1999}\natexlab{}.
\newblock \showarticletitle{On the {K}im and {R}oush voting procedure}.
\newblock \bibinfo{journal}{\emph{Group Decision and Negotiation}}
  \bibinfo{volume}{8} (\bibinfo{year}{1999}), \bibinfo{pages}{109--123}.
\newblock
Issue 2.


\bibitem[\protect\citeauthoryear{Lepelley and Valognes}{Lepelley and
  Valognes}{2003}]%
        {lepelley2003homogeneity}
\bibfield{author}{\bibinfo{person}{Dominique Lepelley} {and}
  \bibinfo{person}{Fabrice Valognes}.} \bibinfo{year}{2003}\natexlab{}.
\newblock \showarticletitle{Voting rules, manipulability and social
  homogeneity}.
\newblock \bibinfo{journal}{\emph{Public Choice}}  \bibinfo{volume}{116}
  (\bibinfo{year}{2003}), \bibinfo{pages}{165--184}.
\newblock
Issue 1.


\bibitem[\protect\citeauthoryear{Mossel, Procaccia, and R{\'a}cz}{Mossel
  et~al\mbox{.}}{2013}]%
        {mossel2013smooth}
\bibfield{author}{\bibinfo{person}{Elchanan Mossel}, \bibinfo{person}{Ariel~D
  Procaccia}, {and} \bibinfo{person}{Mikl{\'o}s~Z R{\'a}cz}.}
  \bibinfo{year}{2013}\natexlab{}.
\newblock \showarticletitle{A smooth transition from powerlessness to absolute
  power}.
\newblock \bibinfo{journal}{\emph{Journal of Artificial Intelligence Research}}
   \bibinfo{volume}{48} (\bibinfo{year}{2013}), \bibinfo{pages}{923--951}.
\newblock


\bibitem[\protect\citeauthoryear{Newman and Ziff}{Newman and Ziff}{2000}]%
        {newman2000efficient}
\bibfield{author}{\bibinfo{person}{Mark Newman} {and} \bibinfo{person}{Robert
  Ziff}.} \bibinfo{year}{2000}\natexlab{}.
\newblock \showarticletitle{Efficient Monte Carlo algorithm and high-precision
  results for percolation}.
\newblock \bibinfo{journal}{\emph{Physical Review Letters}}
  \bibinfo{volume}{85}, \bibinfo{number}{19} (\bibinfo{year}{2000}),
  \bibinfo{pages}{4104}.
\newblock


\bibitem[\protect\citeauthoryear{Peleg}{Peleg}{1979}]%
        {peleg1979largevotingschemes}
\bibfield{author}{\bibinfo{person}{Bezalel Peleg}.}
  \bibinfo{year}{1979}\natexlab{}.
\newblock \showarticletitle{A note on manipulability of large voting schemes}.
\newblock \bibinfo{journal}{\emph{Theory and Decision}}  \bibinfo{volume}{11}
  (\bibinfo{year}{1979}), \bibinfo{pages}{401--412}.
\newblock
Issue 4.


\bibitem[\protect\citeauthoryear{Satterthwaite}{Satterthwaite}{1975}]%
        {satterthwaite1975strategyproofness}
\bibfield{author}{\bibinfo{person}{Mark Satterthwaite}.}
  \bibinfo{year}{1975}\natexlab{}.
\newblock \showarticletitle{Strategy-proofness and {A}rrow's conditions:
  Existence and correspondence theorems for voting procedures and social
  welfare functions}.
\newblock \bibinfo{journal}{\emph{Journal of Economic Theory}}
  \bibinfo{volume}{10}, \bibinfo{number}{2} (\bibinfo{year}{1975}),
  \bibinfo{pages}{187--217}.
\newblock


\bibitem[\protect\citeauthoryear{Slinko}{Slinko}{2002}]%
        {slinko2002classical}
\bibfield{author}{\bibinfo{person}{Arkadii Slinko}.}
  \bibinfo{year}{2002}\natexlab{}.
\newblock \showarticletitle{On asymptotic strategy-proofness of classical
  social choice rules}.
\newblock \bibinfo{journal}{\emph{Theory and Decision}}  \bibinfo{volume}{52}
  (\bibinfo{year}{2002}), \bibinfo{pages}{389--398}.
\newblock
Issue 4.


\bibitem[\protect\citeauthoryear{Walsh}{Walsh}{2010}]%
        {walsh2010empirical}
\bibfield{author}{\bibinfo{person}{Toby Walsh}.}
  \bibinfo{year}{2010}\natexlab{}.
\newblock \showarticletitle{An empirical study of the manipulability of single
  transferable voting}.
\newblock In \bibinfo{booktitle}{\emph{European Conference on Artificial
  Intelligence (ECAI)}}. \bibinfo{pages}{257--262}.
\newblock


\bibitem[\protect\citeauthoryear{Williamson and Sargent}{Williamson and
  Sargent}{1967}]%
        {williamson1967social}
\bibfield{author}{\bibinfo{person}{Oliver Williamson} {and}
  \bibinfo{person}{Thomas Sargent}.} \bibinfo{year}{1967}\natexlab{}.
\newblock \showarticletitle{Social choice: A probabilistic approach}.
\newblock \bibinfo{journal}{\emph{The Economic Journal}} \bibinfo{volume}{77},
  \bibinfo{number}{308} (\bibinfo{year}{1967}), \bibinfo{pages}{797--813}.
\newblock


\bibitem[\protect\citeauthoryear{Xia}{Xia}{2022}]%
        {xia2022impact}
\bibfield{author}{\bibinfo{person}{Lirong Xia}.}
  \bibinfo{year}{2022}\natexlab{}.
\newblock \showarticletitle{The impact of a coalition: Assessing the likelihood
  of voter influence in large elections}.
\newblock \bibinfo{journal}{\emph{arXiv preprint arXiv:2202.06411}}
  (\bibinfo{year}{2022}).
\newblock


\end{thebibliography}


\end{document}